\documentclass[11pt, oneside, reqno]{amsart}

\usepackage{amsmath}
\usepackage{amsthm}
\usepackage{amssymb}

\newcommand{\calB}{\mathcal{B}}

\newcommand{\bbZ}{\mathbb{Z}}

\newcommand{\bbC}{\mathbb{C}}
\newcommand{\bbT}{\mathbb{T}}

\newcommand{\EM}{\mathbf{EM}}
\newcommand{\PEM}{\mathbf{PEM}}

\newcommand{\id}{\operatorname{id}}
\newcommand{\Tr}{\operatorname{Tr}}
\newcommand{\Hom}{\operatorname{Hom}}

\newcommand{\Res}{\operatorname{Res}}

\newcommand{\Ind}{\operatorname{Ind}}
\newcommand{\Span}{\operatorname{span}}
\newcommand{\bra}[1]{\left\langle#1\right|}
\newcommand{\ket}[1]{\left|#1\right\rangle}

\usepackage{thmtools}

\theoremstyle{definition}
\newtheorem{definition}{Definition}[section]

\theoremstyle{plain}
\newtheorem{theorem}[definition]{Theorem}
\newtheorem{proposition}[definition]{Proposition}
\newtheorem{corollary}[definition]{Corollary}
\newtheorem{lemma}[definition]{Lemma}

\newtheorem{question}[definition]{Question}

\theoremstyle{remark}
\newtheorem{remark}[definition]{Remark}
\newtheorem{example}[definition]{Example}

\usepackage{hyperref}
\hypersetup{pdftitle={}}
\usepackage[nameinlink, noabbrev, capitalize]{cleveref}
\usepackage[backend=biber,style=alphabetic]{biblatex}
\usepackage{color}
\usepackage[all]{xy}

\usepackage[a4paper, total={6.5in, 10in}]{geometry}
\usepackage{microtype}

\title[]{Projective error models: Stabilizer codes, Clifford codes, and weak stabilizer codes}

\date{Febraury 25, 2026}
\author[1]{Jonas Eidesen}
\address{Jonas Eidesen, Department of Mathematics, University of Oslo, P.O.Box 1053 Blindern, 0316 Oslo, Norway}
\email{jonaeid@math.uio.no}


\begin{document}

\begin{abstract}
    By defining projective error models we study the mathematical structure of Clifford codes and stabilizer codes using tools from projective representation theory. Furthermore, we introduce a new class of codes which we have called weak stabilizer codes and we determine some relationships between these three classes of codes. We show that the obstruction for a stabilizer code to be non-trivial is given by a class in group cohomology, and we are able to determine similar obstructions for weak stabilizer codes to be non-trivial. In the case where the projective error model corresponds to a nice error basis we give a complete characterization of when a Clifford code is a weak stabilizer code in terms of the size of the group of logical operators and the size of the group of stabilizers of the code. Lastly, we produce two infinite families of Clifford codes that are not stabilizer codes, as well as a method of combining these examples into more examples of non-stabilizer Clifford codes.
\end{abstract}

\maketitle

\begin{center}
    This article is dedicated to the memory of Raymond Laflamme.
\end{center}

\section{Introduction}\label{sec:Introduction}

Quantum error correction is an important component of current scientific efforts towards developing practical realizations of quantum communications. In this article we undertake two central questions in quantum error correction. If we have a physical system that is susceptible to noise we want to try to embed our information into a subsystem, or a code space, that can protect from this noise. The central questions are then:
\begin{enumerate}
    \item Given a code space, what is the noise that this code space can correct?
    \item Given some noise we want to protect against, what is the largest code space that can detect this given noise?
\end{enumerate}
Framing these questions in a mathematically rigorous way can be challenging and is highly dependent on how one chooses to model a physical system. The Knill-Laflamme conditions, cf. \cite{KnillLaflamme97} gives an answer to the first question, which we will describe shortly. An answer to the second question for weak stabilizer codes will be described in \cref{sec:error correction}.

The most standard way to model a quantum mechanical system is with a complex Hilbert space $V$. For our purposes $V$ will always be assumed to be finite dimensional. A quantum mechanical \emph{state} can then be identified with a positive semi-definite operator $\sigma \in \calB(V) := \{ T \colon V \to V : T \text{ is linear} \}$ of trace $1$. Note that any vector $\xi \in V$ of norm $1$ gives rise to a state by considering the rank $1$ projection onto $\xi$, usually denoted $\ket{\xi}\bra{\xi}$. States of this form are called \emph{pure states}. Noise in the system can be modeled by a \emph{quantum channel} $N$ i.e. $N \colon \calB(V) \to \calB(V)$ is a completely positive and trace preserving linear map. With this model, a code space is simply given by a subspace $W$ of $V$. Question (1) above can then be stated as: for which quantum channels $N \colon \calB(V) \to \calB(V)$ does there exist a quantum channel $R \colon \calB(V) \to \calB(V)$ such that
\begin{equation*}
    R \circ N |_{\calB(W)} = \id_{\calB(W)}?
\end{equation*}
Here $\calB(W)$ is embedded into $\calB(V)$ in the natural way. This question is completely answered by the Knill-Laflamme conditions, cf. \cite{KnillLaflamme97}: if $\{K_i\}_{i=1}^{n} \subset \calB(V)$ are Kraus operators of $N$, i.e.
\begin{equation*}
    N(\sigma) = \sum_{i=1}^{n}K_i \sigma K_i^*,
\end{equation*}
and $\sum_{i=1}^{m}K_i^*K_i = \id_V$, then there exists a quantum channel $R$ as above if and only if 
\begin{equation}\label{eqn:KnillLaflamme Conditions 1}
    P_WK_i^*K_jP_W \in \bbC P_W
\end{equation}
for all $i,j \in \{1,2,\dots,n\}$. Here $P_W \in \calB(V)$ denotes the orthogonal projection onto $W$. The containment in \cref{eqn:KnillLaflamme Conditions 1} is usually referred to as the \emph{Knill-Laflamme conditions}.

We will take a slightly different approach. The Knill-Laflamme conditions tells us that it is essential to know which elements $X \in \calB(V)$ satisfy 
\begin{equation}\label{eqn:Knill-Laflamme conditions 2}
    P_W X P_W \in \bbC P_W.
\end{equation}
Elements in $\calB(V)$ that satisfy this condition are called \emph{detectable}. We will consider unitary operators on $V$ that arise as the image of a projectively faithful irreducible projective representation of a finite group. The main question we want to answer is: which of these unitaries are detectable? The reason for this is that a single unitary operator models a single error occurring in the system. To be explicit: if
\begin{equation*}
    U \in U(V):= \{ T \colon V \to V : T \text{ is linear and invertible with } T^{-1} = T^* \},
\end{equation*}
then we can consider the quantum channel $N_{U,p} \colon \calB(V) \to \calB(V)$ defined by
\begin{equation*}
    N_{U,p}(\sigma) = p\sigma + (1 - p) U \sigma U^*,
\end{equation*}
where $p \in (0,1)$ is the probability that no error occurs. Note that this channel is correctable, in the sense that there exists a quantum channel $R_{U,p} \colon \calB(V) \to \calB(V)$ such that
\begin{equation*}
    R_{U,p} \circ N_{U,p} |_{\calB(W)} = \id_{\calB(W)},
\end{equation*}
if and only if the element $U$ is detectable (this follows by the fact that the Knill-Laflamme conditions are preserved under taking adjoints). If $\pi \colon G \to U(V)$ is a projective representation (we will define what this means in \cref{sec:Preliminaries}) where the size of $G$ is $n$, and $p$ is a probability distribution on the set $\{1,\dots,n\}$, then we can define the quantum channel $N_{\pi,p} \colon \calB(V) \to \calB(V)$ by setting
\begin{equation*}
    N_{\pi,p}(\sigma) = \sum_{i = 1}^{n} p(i) \pi(x_i) \sigma \pi(x_i)^*,
\end{equation*}
where $\{x_1,\dots,x_n\}$ is an enumeration of $G$. It follows that the quantum channel $N_{\pi,p}$ is correctable if and only if $\pi(x_i)$ is detectable for every $x_i \in G$.

Generally, the collection of unitaries that are detectable will not form a group. However, the above construction gives a way to construct correctable channels from groups of detectable unitaries.

This article is structured as follows: All the relevant definitions of \emph{projective error models} and of detectable errors will be stated in \cref{sec:EMs and PEMs}. Our definition of a projective error model will essentially be the same as that of a \emph{nice error frame}, cf. \cite{ChienWaldron17}. However, we want to use the tools from projective representation theory to its full potential, and therefore choose to emphasize this aspect of the theory. We give the necessary preliminaries about projective representation theory in \cref{sec:Preliminaries}. In \cref{sec:codes} we define the three classes of codes we are interested in studying: stabilizer codes, Clifford codes, and weak stabilizer codes. We also prove some basic results about how these codes are related to each other. Our definition of a Clifford code generalizes the definition due to Knill, cf. \cite{KnillI96}. We give a description of the detectable errors of Clifford codes in \cref{sec:error correction} in terms of the stabilizers of the code and the logical operators on the code, cf. \cref{thm:detectable errors:Clifford code}. With this we are able to show that the structure of a Clifford code is completely determined by its group of logical operators, cf. \cref{cor:Clifford codes has unique inertia group}. This result is useful tool to determine when a code is a Clifford code. We also give a criteria to test if a code is a weak stabilizer code, cf. \cref{prop:stabilizer codes of whole stabilizer group}. In \cref{sec:weak stabilizer Clifford codes} we give a complete characterization of when a Clifford code is also a weak stabilizer code in the case where the projective error model is a \emph{nice error basis}, cf. \cite{KnillI96}. This recovers similar results by Klappenecker and R{\"o}tteler, and Nicolás, Martínez, and Grassl, cf. \cite{KlappeneckerRötteler04,GrasslMartínezNicolás10}. Some important examples of codes are covered in \cref{sec:non-stabilizer Clifford codes,sec:non-Clifford weak stabilizer codes}. The highlight is \cref{prop:family of non-stabilizer Clifford codes} which gives an infinite family of non-stabilizer Clifford codes in dimensions $2n$, with $n \geq 3$ being odd. We furthermore present a general construction for product codes, cf. \cref{sec:Product codes}. This construction yields infinitely many more examples of Clifford codes that are not stabilizer codes. In the final section, \cref{sec:further questions}, we discuss some natural questions that one could consider for future research.

\subsection*{Acknowledgments}
A lot of this work has been accomplished after many discussions with my supervisors Tron Omland, Erik B{\'e}dos, and Nadia S. Larsen, to whom I am extremely grateful. I would also like to thank Ningping Cao and Alexander Frei for useful discussions.

The first versions of this manuscript was finalized during a research visit at The Institute for Quantum Computing while I was visiting Raymond Laflamme's research group. The research visit was funded by a grant given by The Research Council of Norway.

This research was funded by The Research Council of Norway [Project 345433] and The Norwegian National Security Authority (NSM).

\section{Preliminaries}\label{sec:Preliminaries}

\subsection{Some notation for common groups}

For an abstract cyclic group of $n$ elements we will use the notation $C_n$. Since $C_n$ can faithfully be represented as the $n$th roots of unity, we will also use $C_n$ to denote precisely this, i.e.,
\begin{equation*}
    C_n := \{ z \in \bbC \colon z^n = 1 \}.
\end{equation*}
The group operation in $C_n$ is denoted as multiplication.

We will also use the notation $\bbZ_n$ to denote a cyclic group of $n$ elements. In this case $\bbZ_n$ is \emph{always} the set $\{0,1,2,\dots,n-1\}$. The group operation in $\bbZ_n$ is denoted as addition, and it is implicit that this is done modulo $n$.

We will use the notation $\bbT$ for the circle group, i.e.,
\begin{equation*}
    \bbT := \{ z \in \bbC \colon |z| = 1 \}.
\end{equation*}

If $G$ is a group we will use the notation $1_G$ to denote the identity element in $G$, or simply as $1$ when no confusion can arise. We will also use the notation $Z(G)$ to denote the \emph{center} of the group $G$ i.e.,
\begin{equation*}
    Z(G) := \{ x \in G : xy = yx \text{ for all } y \in G \}.
\end{equation*}

The order, or size, of a group $G$ will be denoted by $|G|$.

\subsection{Projective representation theory: definitions and notation}

We will only work with finite groups in this paper: whenever we write that $G$ is a group, we will assume that $G$ is finite. Furthermore, we will also assume that all representations of groups are unitary and on finite dimensional complex Hilbert spaces. Most of the statements we make about projective representations can be found in for example the book \cite[Chapters 7-9]{CeccheriniSilbersteinTullio22}.

A (linear) representation of a group $G$ on a nonzero Hilbert space $V$ is defined to be a group homomorphism $\lambda \colon G \to U(V)$, ($\lambda$ for linear). A projective representation of $G$ on $V$ is defined to be a function $\pi \colon G \to U(V)$, ($\pi$ for projective), such that the composition $q \circ \pi \colon G \to PU(V)$ is a group homomorphism. Here $PU(V)$ is the quotient of $U(V)$ by scalar multiples of the identity, and $q \colon U(V) \to PU(V)$ is the quotient map. We say that $\pi$ is \emph{projectively faithful} if the composition $q \circ \pi$ is injective. That the composition $q \circ \pi$ is a group homomorphism is equivalent to saying that there exists a function $\sigma \colon G \times G \to \bbT$ such that
\begin{equation*}
    \pi(x)\pi(y) = \sigma(x,y)\pi(xy), \text{ for all } x,y \in G.
\end{equation*}

\begin{remark}
    Note that this definition differs slightly from the one found in \cite{CeccheriniSilbersteinTullio22}, where it is required that
    \begin{equation*}
        \pi(xy) = \sigma(x,y)\pi(x)\pi(y), \text{ for all } x,y \in G.
    \end{equation*}
    This has the effect of conjugating the function $\sigma$ in our definition. In our experience the convention chosen in \cite{CeccheriniSilbersteinTullio22} occurs less frequently in the literature, hence we have made the decision to not use their definition. Our definition agrees with the one found in for example \cite{Cheng15}.
\end{remark}

To make the dependence on the function $\sigma$ explicit, we will often call the function $\pi$ for a $\sigma$-projective representation of $G$ on $V$. One easily shows that $\sigma$ satisfies the following equation for all $x,y,z \in G$:
\begin{equation*}
    \sigma(x,y)\sigma(xy,z) = \sigma(x,yz)\sigma(y,z).
\end{equation*}
This is called the 2-cocycle identity, and we denote the collection of all functions $G \times G \to \bbT$ that satisfy the 2-cocycle identity by $Z^2(G,\bbT)$. Elements in $Z^2(G,\bbT)$ are naturally called \emph{2-cocycles}.

We say that a 2-cocycle $\sigma \in Z^2(G,\bbT)$ is a \emph{2-coboundary} if there exists a function $f \colon G \to \bbT$ such that $\sigma = \delta f$ where $\delta f \colon G \times G \to \bbT$ is the function defined by
\begin{equation*}
    (\delta f)(x,y) = f(x)f(y)\overline{f(xy)}, \text{ for all } x,y \in G.
\end{equation*}
We denote the collection of all 2-coboundaries by $B^2(G,\bbT)$. Note that $Z^2(G,\bbT)$ becomes an abelian group when endowed with the pointwise product, and $B^2(G,\bbT)$ is a subgroup of $Z^2(G,\bbT)$.

We further say that two 2-cocycles $\sigma_1, \sigma_2 \in Z^2(G,\bbT)$ are \emph{cohomologous} if there exists a function $f \colon G \to \bbT$ such that $\sigma_1 = (\delta f) \, \sigma_2$. This defines an equivalence relation on $Z^2(G,\bbT)$, and we denote the collection of all equivalence classes under this equivalence for $H^2(G,\bbT)$. Note that
\begin{equation*}
    H^2(G,\bbT) = Z^2(G,\bbT) / B^2(G,\bbT),
\end{equation*}
which makes it clear that $H^2(G,\bbT)$ is an abelian group. For a 2-cocycle $\sigma \in Z^2(G,\bbT)$ we will use the notation $[\sigma]$ to denote its class in $H^2(G,\bbT)$. If the 2-cocycle $\sigma$ arises from a projective representation $\pi$ of $G$, then we will sometimes use the notation $[\sigma_\pi]$ to denote the class of $\sigma$ in $H^2(G,\bbT)$. A well-known but nontrivial fact is that $H^2(G,\bbT)$ is finite whenever $G$ is finite. We will make use of this fact in \cref{prop:correspondence between EM and PEM}.

It is at this point worth mentioning that all of this terminology has a natural place in what is called \emph{group cohomology}. We will need very few of the general tools from group cohomology, but the interested reader may consult \cite{Brown82}. Note that the connection between projective representation theory and group cohomology goes back to the late 1940s, while projective representation theory for finite groups was developed by Schur in the early 1900s. See \cite{Packer08} for a survey of projective representation theory where all these facts are mentioned.

If $\pi_i \colon G \to U(V_i)$, $i = 1,2$, are two $\sigma$-projective representations of $G$, we define an \emph{intertwiner} from $\pi_1$ to $\pi_2$ to be a linear map $T \colon V_1 \to V_2$ such that the diagram
\begin{equation*}
    \xymatrix{
        V_1 \ar[d]_-{\pi_{1}(x)} \ar[r]^-{T} & V_2 \ar[d]^-{\pi_{2}(x)} \\
        V_1 \ar[r]_-{T} & V_2
    }
\end{equation*}
commutes for all $x \in G$. We will often use the notation $T \colon \pi_1 \to \pi_2$ to emphasize that $T$ is an intertwiner and not only a linear map. The collection of all intertwiners between $\pi_1$ and $\pi_2$ is denoted by $\Hom_G(\pi_1,\pi_2)$, and it is clear that this has the natural structure of a complex vector space. We say that $\pi_1$ is isomorphic to $\pi_2$ if there exists an intertwiner $T \in \Hom_G(\pi_1,\pi_2)$ that is also a bijection. In the case where $\pi_1$ is isomorphic to $\pi_2$ we will sometimes write $\pi_1 \simeq \pi_2$.

We have that the collection of $\sigma$-projective representations of $G$ along with intertwiners form a category. This category can be identified with modules over $\bbC[G]^\sigma$, where $\bbC[G]^\sigma$ denotes the \emph{$\sigma$-twisted group algebra of $G$}. As a complex vector space it is defined to be the collection of functions $F \colon G \to \bbC$. For $y \in G$, let $\delta_y \colon G \to \bbC$ be defined by
\begin{equation*}
    \delta_y(x) =
    \begin{cases}
        1 & \text{if } x = y, \\
        0 & \text{if } x \neq y.
    \end{cases}
\end{equation*}
Then $\{\delta_y\}_{y \in G}$ form a basis for $\bbC[G]^\sigma$. Sometimes, when it is clear from context, we will simply use the notation $y = \delta_y$. For two functions $F_1,F_2 \in \bbC[G]^\sigma$, we define their product by the following $\sigma$-twisted convolution product:
\begin{equation*}
    (F_1 *_\sigma F_2)(x) = \sum_{y \in G} \sigma(y,y^{-1}x) F_1(y)F_2(y^{-1}x).
\end{equation*}
One confirms, with some effort, that this turns $\bbC[G]^\sigma$ into an algebra. In the special case where $F_2 = y$ we have that
\begin{equation*}
    (F_1 *_\sigma y)(x) = \sigma(xy,y)F_1(xy).
\end{equation*}
Similarly, when $F_1 = y$ we have that
\begin{equation*}
    (y *_\sigma F_2)(x) = \sigma(y,y^{-1}x)F_2(y^{-1}x).
\end{equation*}

We can further endow $\bbC[G]^\sigma$ with an inner product defined by
\begin{equation}\label{eqn:inner product}
    \langle F_1, F_2 \rangle = \frac{1}{|G|} \sum_{x \in G}F_1(x) \overline{F_2(x)}, \text{ for } F_1,F_2 \in \bbC[G]^\sigma.
\end{equation}
With this $\bbC[G]^\sigma$ may also be regarded as a Hilbert space. Note that this Hilbert space is usually denoted by $\ell^2(G)$. We will not use this notation, as the main thing we care about is the fact that left multiplication in $\bbC[G]^\sigma$ by the basis vectors $\{\delta_y\}_{y \in G}$ induces a $\sigma$-projective representation of $G$ (namely the left-regular $\sigma$-projective representation of $G$). The above inner product ensures that this projective representation is unitary. We will however not explicitly need that this projective representation is unitary, hence we choose to emphasize the algebraic structure involved by only using the notation $\bbC[G]^\sigma$.

If we have two 2-cocycles $\sigma_1,\sigma_2 \in Z^2(G,\bbT)$ and a function $f \colon G \to \bbT$, then the linear map $L_f \colon \bbC[G]^{\sigma_1} \to \bbC[G]^{\sigma_2}$ defined by
\begin{equation*}
    (L_f(F))(x) = f(x) \, F(x),
\end{equation*}
is an algebra isomorphism if and only if $\sigma_1 = (\delta f) \, \sigma_2$. In particular this means that the $\sigma$-projective representation theory of $G$ is only dependent on the class $[\sigma] \in H^2(G, \bbT)$, which is a fact that is very useful when one is doing concrete computations as one can often find 2-cocycles with very nice properties that make computations easier.

\subsection{Projective representation theory: irreducible representations}

If $\pi$ is a projective representation of a group $G$ on a Hilbert space $V$, we say that a subspace $W$ of $V$ is \emph{$\pi$-invariant} if $\pi(x)W \subset W$ for all $x \in G$. Then $\pi$ is said to be \emph{irreducible} if the only $\pi$-invariant subspaces of $V$ are $V$ and $\{0\}$. We will often make use of the fact that $\pi$ is irreducible if and only if
\begin{equation*}
    \Span\{\pi(x) \colon x \in G \} = \calB(V).
\end{equation*}

Projective representations enjoy analogs of Maschke's Theorem and Schur's Lemma, cf. \cite[Corollary 7.15]{CeccheriniSilbersteinTullio22}. Maschke's Theorem says that any projective representation $\pi$ of $G$ is isomorphic to a direct sum of irreducible projective representations. Schur's Lemma states that if $\pi_1$ and $\pi_2$ are irreducible projective representations of $G$, then
\begin{equation*}
    \dim \Hom_G(\pi_1,\pi_2) =
    \begin{cases}
        1 & \text{if } \pi_1 \simeq \pi_2, \\
        0 & \text{if } \pi_1 \not\simeq \pi_2.
    \end{cases}
\end{equation*}

Projective representations also enjoy a character theory which can be used to make powerful arguments. We will need the following result, see \cite[Proposition 2.2]{Cheng15}. If $\pi_1$ and $\pi_2$ are $\sigma$-projective representations of $G$ (not necessarily irreducible), then
\begin{equation}\label{eqn:inner porduct of characters}
    \langle \chi_{\pi_1}, \chi_{\pi_2} \rangle = \dim \Hom_G(\pi_1,\pi_2).
\end{equation}
Here the inner product is the one defined in \cref{eqn:inner product}, and for $i = 1,2$, $\chi_{\pi_i} \in \bbC[G]^\sigma$ are the functions defined by
\begin{equation*}
    \chi_{\pi_i}(x) = \Tr(\pi_i(x)).
\end{equation*}

\subsection{Projective representation theory: induction and Clifford theory}

Suppose that $H$ is a subgroup of $G$ and $\pi$ is a $\sigma$-projective representation of $G$ on $V$. We define a projective representation $\Res^G_H \pi$ of $H$ on $V$ by letting
\begin{equation*}
    \Res^G_H\pi(x) = \pi(x), \text{ for } x \in H.
\end{equation*}
We have that $\Res^G_H \pi$ is $\Res^G_H\sigma$-projective, where $\Res^G_H\sigma \in Z^2(H,\bbT)$ is defined by restricting $\sigma$ to $H \times H$. 

Induction is a procedure that is adjoint to restriction. Concretely: if $\theta$ is a $\Res^G_H\sigma$-projective representation of $H$ on $W$, we define a $\sigma$-projective representation $\Ind_H^G\theta$ of $G$ on $\Ind_H^G(W)$, where $\Ind_H^G(W)$ is the balanced tensor product:
\begin{equation*}
    \Ind_H^G(W) := \bbC[G]^\sigma \otimes_H W.
\end{equation*}
The subscript $H$ indicates that for any $F \in \bbC[G]^\sigma$, $\xi \in W$, and $x \in H$ we have that
\begin{equation*}
    F \otimes (\theta(x) \xi) = (F *_\sigma x) \otimes \xi.
\end{equation*}
If $\{e_i\}_{i = 1}^n$ is a basis for $W$, this property ensures that $\{ r \otimes e_i : r \in G/H, \ i \in \{1,\dots,n\} \}$ is a basis for $\Ind_H^G(W)$.
Then for any $x \in G$ we define $\Ind_H^G\theta(x)$ by linearly extending the following identity:
\begin{equation*}
    (\Ind_H^G\theta(x))(F \otimes \xi) = (x *_\sigma F) \otimes \xi, \text{ for } F \in \bbC[G]^\sigma, \xi \in W.
\end{equation*}
Since the $\sigma$-twisted convolution product in $\bbC[G]^\sigma$ is associative, $\Ind^G_H\theta \colon G \to U(\Ind^G_H(W))$ indeed defines a $\sigma$-projective representation of $G$.

Induced representations satisfy Frobenius reciprocity: there is a natural isomorphism
\begin{equation*}
    \Hom_G(\Ind_H^G\theta ,\pi) \simeq \Hom_H(\theta, \Res^G_H\pi)
\end{equation*}
for any $\sigma$-projective representation $\pi$ of $G$ and any $\Res^G_H\sigma$-projective representation $\theta$ of $H$, cf. \cite[Theorem 8.15]{CeccheriniSilbersteinTullio22}.

We will need to introduce some more notation before we can state the results we need from Clifford theory. Assume now that $H$ is a normal subgroup of $G$, that $\pi$ is an irreducible $\sigma$-projective representation of $G$, and that $\theta$ is an irreducible $\Res^G_H\sigma$-projective representation of $H$. Define the \emph{inertia group of $\theta$} to be
\begin{equation*}
    I_{G}(\theta) := \{ x \in G \colon \theta^x \simeq \theta \},
\end{equation*}
where $\theta^x$ is the $\Res^G_H\sigma$-projective representation of $H$ given by
\begin{equation*}
    \theta^x(y) = \sigma(x^{-1},y) \overline{\sigma(x^{-1}yx,x^{-1})} \theta(x^{-1}yx), \text{ for } y \in H.
\end{equation*}
Then the following version of Mackey's Lemma holds: If $\ell := \dim \Hom_H(\theta,\Res^G_H\pi)$ is nonzero, then
\begin{equation*}
    \Res^G_H\pi \simeq \ell \bigoplus_{r \in G/I_{G}(\theta)} \theta^r,
\end{equation*}
cf. \cite[Theorem 9.7]{CeccheriniSilbersteinTullio22}.

Finally, Clifford correspondence says the following: if $\ell = \dim\Hom_H(\theta,\Res^G_H\pi)$ is nonzero, then there exists a unique irreducible $\Res^G_{I_G(\theta)}\sigma$-projective representation $\rho$ of $I_G(\theta)$ such that $\Hom_H(\theta,\Res^{I_G(\theta)}_H\rho)$ is nonzero and
\begin{equation*}
    \Ind_{I_G(\theta)}^G\rho \simeq \pi,
\end{equation*}
Furthermore, $\dim\Hom_H(\theta,\Res^{I_G(\theta)}_H\rho) = \ell$, cf. \cite[Theorem 9.12]{CeccheriniSilbersteinTullio22}.

\section{Error models and projective error models}\label{sec:EMs and PEMs}

If $V$ is a nonzero Hilbert space we want to consider any unitary operator on $V$ as a potential error that can occur in the system that is modeled by $V$. It is however useful to instead work with some finite amount of data instead of the whole group of unitaries $U(V)$. In this section we make precise what type of finite data of $U(V)$ we want to work with. We will also always assume that the Hilbert spaces we consider are nonzero.

\begin{definition}\label{def:error model}
    Let $V$ be a Hilbert space. If $E$ is a group admitting a faithful irreducible representation on $V$ and $\lambda$ is a choice of such a representation we say that the pair $(E,\lambda)$ is an \emph{error model on $V$}. The class of all such pairs is denoted $\EM_V$.
\end{definition}

If $(E,\lambda) \in \EM_V$ then for any $x \in Z(E)$ we have that $\lambda(x) \in \bbT 1_V$ by Schur's Lemma. Since any quantum mechanical measurement cannot detect a global phase change, all of the elements in $Z(E)$ are unnecessary data. Hence, we may consider the quotient of $E$ by $Z(E)$ instead, or abstractly we may define a \emph{projective error model}:

\begin{definition}\label{def:projective error model}
    Let $V$ be a Hilbert space. If $G$ is a group admitting a projectively faithful irreducible projective representation on $V$ and $\pi$ is a choice of such a representation we say that the pair $(G,\pi)$ is a \emph{projective error model on $V$}. The class of all such pairs is denoted $\PEM_V$.
\end{definition}

The first question we would like to ask is what type of groups admit faithful irreducible representations, and what groups admit projectively faithful irreducible projective representations, and what are the connections between the two? For finite groups, the first question is essentially answered in \cite{Gaschtz54}. For the whole story, in a bit more general setting, one can see \cite{BekkadelaHarpe08,BekkadelaHarpe13}. Although it is useful to know that one can essentially classify these groups, we would like to briefly outline how the first class of groups is connected to the second class of groups.

\begin{proposition}\label{prop:correspondence between EM and PEM}
    Let $V$ be a Hilbert space. If $(E,\lambda) \in \EM_V$ then the quotient $G := E/Z(E)$ admits a projectively faithful irreducible projective representation $\pi$ on $V$. Furthermore, the group homomorphism $q \circ \pi$ is uniquely determined by the representation $\lambda$.

    Conversely, if $(G',\pi') \in \PEM_V$, then for any natural number $n$ such that the order of $[\sigma_{\pi'}] \in H^2(G_2,\bbT)$ divides $n$, we have that there exists a central extension $E'$ of $G'$ by $C_n$ that admits a faithful irreducible representation on $V$.
\end{proposition}
\begin{proof}
    For the first statement, let $(E,\lambda) \in \EM_V$. We will define a projective representation of $G$ on $V$. For each $x \in G$ let $s_x \in E$ be any coset representative of $x$. For $x,y \in G$ we have that there exists a unique element $\alpha(x,y) \in Z(E)$ such that $s_x s_y = \alpha(x,y) s_{xy}$. Since $\lambda$ is an irreducible representation of $E$ we have by Schur's Lemma that
    \begin{equation*}
        \lambda(\alpha(x,y)) = \sigma(x,y)1_{U(V)}
    \end{equation*}
    for some $\sigma(x,y) \in \bbT$. Letting $x$ and $y$ range over all elements in $G$ we obtain a 2-cocycle $\sigma \in Z^2(G,\bbT)$. We now define a function $\pi \colon G \to U(V)$ by the equation
    \begin{equation*}
        \pi(x) := \lambda(s_x).
    \end{equation*}
    Then $\pi$ is a $\sigma$-projective representation of $G$. Note that by construction we have that the diagram
    \begin{equation*}
        \xymatrix{
            E \ar[r]^-{p} \ar[d]_-{\lambda_1} & G \ar[d]^-{q \circ \pi_1} \\
            U(V) \ar[r]^-{q} & PU(V)
        }
    \end{equation*}
    commutes. By surjectivity of the quotient map $p \colon E \to G$ we get that $q \circ \pi$ is the unique group homomorphism making this diagram commute.
    
    It remains to check that $\pi$ is projectively faithful: Let $x \in \ker{(q \circ \pi)}$. By assumption we have that 
    \begin{equation*}
        \lambda(s_x) = \pi(x) = z 1_{U(V)}
    \end{equation*}
    for some $z \in \bbT$. Hence, for any $g \in E$ we have that
    \begin{equation*}
        \lambda(s_x g) = \lambda(g s_x).
    \end{equation*}
    Since $\lambda$ is injective ($\lambda$ is a faithful representation of $E$ on $V$) we have that $s_x g = g s_x$ for every $g \in E$. Thus, $s_x \in Z(E)$, meaning that $x = 1_{G}$. Hence, $q \circ \pi$ is injective.

    For the second statement, let now $(G',\pi') \in \PEM_V$, and suppose that $\pi'$ is a $\sigma$-projective representation of $G'$ on $V$. Let $n$ be a natural number such that the order of $[\sigma] \in H^2(G',\bbT)$ divides $n$. The following short exact sequence,
    \begin{equation*}
        \xymatrix{
            1 \ar[r] & C_n \ar[r]^-{i} & \bbT \ar[r]^-{(-)^n} & \bbT \ar[r] & 1,
        }
    \end{equation*}
    induces a long exact sequence in group cohomology with the following segment,
    \begin{equation*}
        \xymatrix{
            \cdots \ar[r] & H^2(G',C_n) \ar[r]^-{i^*} & H^2(G',\bbT) \ar[r]^-{n} & H^2(G',\bbT) \ar[r] & \cdots,
        }
    \end{equation*}
    where $n$ here denotes the map that is multiplication by $n$. The fact that the order of $[\sigma]$ divides $n$ is precisely the fact that $[\sigma] \in \ker{n}$, and since this sequence is exact, this means that there exists an element $[\sigma'] \in H^2(G',C_n)$ such that $i^*([\sigma']) = [\sigma]$. Fix a representative $\sigma' \colon G' \times G' \to C_n$ of $[\sigma']$ and choose a function $f \colon G' \to \bbT$ such that $\sigma' = (\delta f) \, \sigma$. Then we have that
    \begin{equation*}
        f(x)\pi'(x) \, f(y)\pi'(y) = \sigma'(x,y) f(xy)\pi'(xy) \text{ for all } x,y \in G'.
    \end{equation*}
    Let $E'$ denote the central extension of $G'$ by $C_n$ corresponding to the element $[\sigma'] \in H^2(G',C_n)$. Then define a representation $\lambda' \colon E' \to U(V)$ by
    \begin{equation*}
        \lambda'(z,x) = z f(x)\pi'(x)
    \end{equation*}
    for $z \in C_n$ and $x \in G'$. This is easily seen to define a faithful irreducible representation of $E'$ on $V$, completing the proof.
\end{proof}

Given an error model $(E,\lambda) \in \EM_V$ and a code space $W \subset V$, one of the most important quantities to study is the collection of detectable errors in $E$ over $W$. This will be the elements of $E$ that satisfy the Knill-Laflamme conditions, \cref{eqn:Knill-Laflamme conditions 2}. To be precise we make the following definitions.

\begin{definition}\label{def:detectable errors}
    Let $V$ be a Hilbert space and $W \subset V$ be a subspace. If $(E,\lambda) \in \EM_V$ we define the set of detectable errors in $E$ over $W$ to be the set
    \begin{equation*}
        D_{(E,\lambda)}(W) := \{ x \in E \colon P_W \lambda(x) P_W \in \bbC P_W \},
    \end{equation*}
    where $P_W$ denotes the orthogonal projection onto $W$. If $(G,\pi) \in \PEM_V$ we analogously define the set
    \begin{equation*}
        D_{(G,\pi)}(W) = \{ x \in G \colon P_W \pi(x) P_W \in \bbC P_W \}.
    \end{equation*}
\end{definition}

With this we can have a precise formulation of the questions we want to answer. Given an error model $(E,\lambda) \in \EM_V$ we want to answer: 
\begin{enumerate}
    \item If $W \subset V$ is a subspace, can we gain an alternative description of $D_{(E,\lambda)}(W)$?
    \item If $D \subset G$ is some subset, what is the largest (in terms of dimension) subspace $W$ of $V$ such that $D \subset D_{(E,\lambda)}(W)$, and can we describe this subspace $W$?
\end{enumerate}
Of course we are also interested in the analogous questions for projective error models, and in some sense, we are only interested in these questions for projective error models:

\begin{proposition}\label{prop:detectable errors of EM vs PEM}
    Let $V$ be a Hilbert space and $W \subset V$ be a subspace. If $(E,\lambda) \in \EM_V$ then
    \begin{equation*}
        D_{(E,\lambda)}(W) = \{ x \in E \colon p(x) \in D_{(E/Z(E),\pi)}(W) \},
    \end{equation*}
    where $p \colon E \to E/Z(E)$ denotes the quotient map, and $\pi \colon E/Z(E) \to U(V)$ is any projectively faithful irreducible projective representation such that $q \circ \pi \circ p = q \circ \lambda$.
\end{proposition}
\begin{proof}
    Let $(E,\lambda) \in \EM_V$ and $\pi \colon E/Z(E) \to U(V)$ be any projectively faithful irreducible projective representation such that $q \circ \pi \circ p = q \circ \lambda$. Note that the existence of such a projective representation $\pi$ is guaranteed by \cref{prop:correspondence between EM and PEM}. Then for any $x \in E$ we have that there exists an element $f(x) \in \bbT$ such that
    \begin{equation*}
        \pi(p(x)) = f(x)\lambda(x).
    \end{equation*}
    From this it is clear that $x \in D_{(E,\lambda)}(W)$ if and only if $p(x) \in D_{(E/Z(E),\pi)}(W)$, completing the proof.
\end{proof}

Note that this also shows that for any projective error model $(G,\pi) \in \PEM_V$, and any subspace $W \subset V$, the collection of detectable errors $D_{(G,\pi)}(W)$ is only dependent on the group homomorphism $q \circ \pi \colon G \to PU(V)$, and not on the projective representation $\pi$ of $G$ on $V$. One could therefore consider defining a projective error model on $V$ to be a pair $(G, \pi')$ where $G$ is a finite group and $\pi' \colon G \to PU(V)$ is a group homomorphism satisfying some niceness properties analogous to irreducibility and faithfulness. The problem with this is that one does not have a good notion of a representation theory for such group homomorphisms. The reason for this is that there might be inequivalent projective representations $\pi_1, \pi_2 \colon G \to U(V)$ such that
\begin{equation*}
    q \circ \pi_1 = q \circ \pi_2 = \pi'.
\end{equation*}
This happens for example for the group $D_4$, the dihedral group of $8$ elements, which has precisely two irreducible projective representations of degree 2 (with a nontrivial 2-cocycle). Both of these representations give the same group homomorphism when precomposed with the quotient map as above. A way to see that this is the case is by looking at the character table for $D_4 = \langle a, b \ | \ a^4 = b^2 = 1, \ bab = a^{-1} \rangle$:
\begin{equation*}
    \begin{array}{|c||r||r|r||r||r|r||r|r|}
        \hline
        \rm Elements & 1 & a & a^3 & a^2 & b & a^2b & ab & a^3b \cr
        \hline \hline
        \rho_{1} & 1 & 1 & 1 & 1 & 1 & 1 & 1 & 1 \cr\hline
        \rho_{2} & 1 & 1 & 1 & 1 & -1 & -1 & -1 & -1 \cr\hline
        \rho_{3} & 1 & -1 & -1 & 1 & 1 & 1 & -1 & -1 \cr\hline
        \rho_{4} & 1 & -1 & -1 & 1 & -1 & -1 & 1 & 1 \cr\hline
        \rho_{5} & 2 & 0 & 0 & -2 & 0 & 0 & 0 & 0 \cr\hline\hline
        \chi_{1} & 2 & 1 + i & 1 - i & 0 & 0 & 0 & 0 & 0 \cr\hline
        \chi_{2} & 2 & -1 - i & -1 + i & 0 & 0 & 0 & 0 & 0 \cr\hline
    \end{array}
\end{equation*}
Here we have grouped the elements of $D_4$ according to their conjugacy classes, the $\rho_i$'s are all the linear characters of $D_4$, and $\chi_1$ and $\chi_2$ are the two projective characters of $D_4$, where $\chi_1$ is the character associated to the projective representation described in \cref{prop:projective XP-error model}. Let $\pi_1 \colon D_4 \to U(\bbC^2)$ be a projective representation whose character is $\chi_1$. Since the degree of $\rho_3$ above is $1$, we have that $\rho_3$ is in fact a group homomorphism, hence the projective representation $\pi_2 \colon D_4 \to U(\bbC^2)$ defined by
\begin{equation*}
    \pi_2(x) = \rho_3(x)\pi_1(x)
\end{equation*}
has the same 2-cocycle associated with it as $\pi_1$. Clearly we have that $q \circ \pi_1 = q \circ \pi_2$ since $\rho_3(x) \in \bbT$ for every $x \in D_4$. But the character of $\pi_2$ is necessarily given by $\rho_3 \, \chi_1 = \chi_2$, hence $\pi_1$ and $\pi_2$ are not isomorphic.

Next we present the most well-known example of an error model.

\begin{example}\label{ex:Pauli group}
    Consider the Pauli group $P_n$ on $n$ qubits given as a subgroup of $U((\bbC^2)^{\otimes n})$. It is an error model, in the sense of \cref{def:error model}, by considering the pair $(P_n,i)$, where $i \colon P_n \to U((\bbC^2)^{\otimes n})$ is the inclusion homomorphism. Physically, the Pauli group is an error model that models errors occurring independently on different qubits, this is typical of product error models which we will get back to.
    
    A projective error model associated to the Pauli error model is given by $(G,\pi)$ where $G = (\bbZ_2 \times \bbZ_2)^{n}$ and $\pi$ is the $n$-fold product representation of the only projectively faithful irreducible projective representation of $\bbZ_2 \times \bbZ_2$. Explicitly, if
    \begin{equation*}
        X =
        \begin{bmatrix}
            0 & 1 \\
            1 & 0
        \end{bmatrix}, \ 
        Z =
        \begin{bmatrix}
            1 & 0 \\
            0 & -1 
        \end{bmatrix},
    \end{equation*}
    we define $\pi \colon (\bbZ_2 \times \bbZ_2)^{n} \to U((\bbC^2)^{\otimes n})$ by
    \begin{equation*}
        \pi(a_1,b_1,\dots,a_n,b_n) = X^{a_1} Z^{b_1} \otimes \cdots \otimes X^{a_n}Z^{b_n}.
    \end{equation*}
\end{example}

The Pauli group is a very important example of an error model. The first versions of more general error models are due to Knill, cf. \cite{KnillI96,KnillII96} where the term \emph{nice error basis} is introduced. In our setup, a nice error basis is a projective error model $(G,\pi) \in \PEM_V$ such that $|G| = (\dim V)^2$. In the literature, groups that admit an irreducible projective representation with this property are sometimes called \emph{groups of central type}. The reader should be warned that this term is sometimes also used to refer to a group $E$ such that the quotient $E/Z(E)$ admits an irreducible projective representation on a Hilbert space $V$ with the property that $|E/Z(E)| = (\dim V)^2$.

In \cite{ChienWaldron17} the notion of a \emph{nice error frame} is introduced. This is precisely the data of a projective error model as above. For a more thorough account of nice error frames, one can consult the book \cite{Waldron18}.

\section{Examples of projective error models} \label{sec:EM examples}

In this section we will give some important examples of some of the concepts we have talked about so far. First we would like to point out that for any Hilbert space $V$, the collection of projective error models $\PEM_V$ is non-empty. This is of course well known, but it is nice to have it highlighted.

\begin{proposition}\label{prop:projective error models exist}
    Let $n \geq 1$ be an integer, $\zeta_n \in \bbC$ be a primitive $n$-th root of unity, and consider the group $\bbZ_n \times \bbZ_n$. Then the function $\pi \colon \bbZ_n \times \bbZ_n \to U(\bbC^n)$ defined by
    \begin{equation*}
        \pi(a,b) = X_n^a Z_n^b,
    \end{equation*}
    where
    \begin{equation*}
        X_n =
        \begin{bmatrix}
            0 & 1 & 0 & \cdots & 0 & 0 \\
            0 & 0 & 1 & \cdots & 0 & 0 \\
            0 & 0 & 0 & \cdots & 0 & 0 \\
            \vdots & \vdots & \vdots & \ddots & \vdots & \vdots \\
            0 & 0 & 0 & \cdots & 0 & 1 \\
            1 & 0 & 0 & \cdots & 0 & 0
        \end{bmatrix},
        \text{ and }
        Z_n =
        \begin{bmatrix}
            1 & 0 & 0 & \cdots & 0 & 0 \\
            0 & \zeta_n & 0 & \cdots & 0 & 0 \\
            0 & 0 & \zeta_n^2 & \cdots & 0 & 0 \\
            \vdots & \vdots & \vdots & \ddots & \vdots & \vdots \\
            0 & 0 & 0 & \cdots & \zeta_n^{n-2} & 0 \\
            0 & 0 & 0 & \cdots & 0 & \zeta_n^{n-1}
        \end{bmatrix},
    \end{equation*}
    is a projectively faithful irreducible projective representation of $\bbZ_n \times \bbZ_n$ on $\bbC^n$.
\end{proposition}
\begin{proof}
    One computes that $(X_n)^n = (Z_n)^n = I_n$ and that $X_nZ_n = \zeta_n Z_nX_n$. Hence, the defining relations for $\bbZ_n \times \bbZ_n$ are satisfied by the elements $q(X_n),q(Z_n) \in PU(\bbC^n)$. Thus, the composition $q \circ \pi \colon \bbZ_n \times \bbZ_n \to PU(\bbC^n)$ is a group homomorphism, meaning that $\pi$ is a projective representation of $\bbZ_n \times \bbZ_n$ on $\bbC^n$.

    Notice that the set $\{\pi(a,b) \colon a,b \in \bbZ_n \}$ is linearly independent, so since the cardinality of $\bbZ_n \times \bbZ_n$ is $n^2$ we get that this is in fact a basis for $M_n(\bbC)$. This shows that $\pi$ is projectively faithful and irreducible, completing the proof.
\end{proof}

The next class of projective error models is derived from the XP-stabilizer formalism, cf. \cite{Webster2022}. We would also like to mention that the following facts about the dihedral groups are well known, and an exposition can be found in \cite[Chapter 3.7]{Karpilovsky85}. For the ease of the reader we will present a proof here.

\begin{proposition}\label{prop:projective XP-error model}
    Let $n \geq 2$ be an integer, $\zeta_n \in \bbC$ be a primitive $n$-th root of unity, and consider the dihedral group of $2n$ elements with the following presentation:
    \begin{equation*}
        D_n = \langle a,b \, \vert \, a^n = b^2 = 1, \, bab = a^{n-1} \rangle = \{ 1, a, a^2, \dots, a^{n-1}, b, ba, \dots, ba^{n-1} \}.
    \end{equation*}
    Then the function $\pi \colon D_n \to U(\bbC^2)$ defined by
    \begin{equation*}
        \pi(b^ka^l) = X^k P^l, \text{ for } k = 0,1 \text{ and } l = 0,1,\dots,n-1,
    \end{equation*}
    where 
    \begin{equation*}
        X =
        \begin{bmatrix}
            0 & 1 \\
            1 & 0
        \end{bmatrix},
        \text{ and }
        P =
        \begin{bmatrix}
            1 & 0 \\
            0 & \zeta_n
        \end{bmatrix},
    \end{equation*}
    is a projectively faithful irreducible projective representation of $D_n$ on $\bbC^2$.
\end{proposition}
\begin{proof}
    We first compute that $P^n = X^2 = I$ and that $XPX = \zeta_n P^{n-1}$. This means that the defining relations for $D_n$ are satisfied by the elements $q(P), q(X) \in PU(\bbC^2)$. Thus, the composition $q \circ \pi \colon D_n \to PU(\bbC^2)$ is a group homomorphism, meaning that $\pi$ is a projective representation of $D_n$ on $\bbC^2$.

    Notice further that the elements $I, X, P, XP$ are linearly independent, hence they form a basis for $M_2(\bbC)$. In particular this means that $\Span\{\pi(x) \colon x \in D_n\} = M_2(\bbC)$, hence $\pi$ is irreducible.

    Finally, one easily computes that $\ker(q \circ \pi) = \{ 1 \}$, hence $\pi$ is projectively faithful, completing the proof.
\end{proof}

\begin{remark}
    One can compute that the 2-cocycle associated to the above projective representation is given by the function $\sigma \colon D_n \times D_n \to \bbT$, defined by
    \begin{equation*}
        \sigma(b^{k_1}a^{l_1}, b^{k_2}a^{l_2}) = \zeta_n^{k_2 l_1}.
    \end{equation*}
    Furthermore, the second cohomology group
    \begin{equation*}
        H^2(D_n, \bbT) = \begin{cases}
            0, & n \text{ is odd}, \\
            \bbZ_2, & n \text{ is even}.
        \end{cases}
    \end{equation*}
    Hence, when $n$ is odd, $\sigma$ is a coboundary. In fact, if we define $f \colon D_n \to \bbT$ by
    \begin{equation*}
        f(b^ka^l) = \zeta_n^{(n-1)l/2},
    \end{equation*}
    then $\sigma \, (\delta f) = 1$.
\end{remark}

Given two projective error models, there is a natural way of combining them to give us a new projective error model. In what follows we explain how to construct such product projective error models.

\begin{proposition}\label{prop:product error models}
    Let $V_1$ and $V_2$ be Hilbert spaces and suppose that $(G_1, \pi_1) \in \PEM_{V_1}$ and $(G_2, \pi_2) \in \PEM_{V_2}$. Then the projective representation $\pi_1 \otimes \pi_2 \colon G_1 \times G_2 \to U(V_1 \otimes V_2)$ defined by
    \begin{equation*}
        (\pi_1 \otimes \pi_2)(x,y) = \pi_1(x) \otimes \pi_2(y), \ x \in G_1, \text{ and } y \in G_2,
    \end{equation*}
    is a projectively faithful irreducible projective representation of $G_1 \times G_2$ on $V_1 \otimes V_2$.
\end{proposition}
\begin{proof}
    It is straightforward to check that the projective representation $\pi_1 \otimes \pi_2$ is irreducible. One can for example confirm that $\Span\{ (\pi_1 \otimes \pi_2)(x,y) \colon x \in G_1, \ y \in G_2 \} = B(V_1 \otimes V_2)$.

    To see that $\pi_1 \otimes \pi_2$ is faithful, suppose that $(\pi_1 \otimes \pi_2)(x,y)$ is a scalar multiple of the identity. Then both $\pi_1(x)$ and $\pi_2(y)$ are scalar multiples of the identities on $V_1$ and $V_2$ respectively. Since both $\pi_1$ and $\pi_2$ are assumed to be projectively faithful we can conclude that $x = 1_{G_1}$ and $y = 1_{G_2}$. This completes the proof.
\end{proof}

This fact is also well known; see for example \cite[Theorem 10]{Serre77} for a proof in the linear setting.

With this established we see that the projective Pauli error model, as explained in \cref{ex:Pauli group}, arise by considering either of the projective error models in \cref{prop:projective error models exist} or \cref{prop:projective XP-error model} with $n = 2$, and taking products of these models with themselves as explained in \cref{prop:product error models}.

When taking products as above we can also incorporate an action from the symmetric group on $n$ letters, $S_n$. This gives a more sophisticated way of constructing projective error models, as has been done in for example \cite{Pollatsek-Ruskai04} and \cite{AydinAlekseyevBarg24} where the authors construct permutation invariant codes that can correct errors from the Pauli group.

\begin{proposition}\label{prop:permutation product PEM}
    Let $V$ be a Hilbert space and suppose that $(G,\pi) \in \PEM_V$. For $n \geq 1$ let $S_n$ act on $G^n$ by permuting the $n$ copies of $G$. Then the semi direct product $G^n \rtimes S_n$ constructed from this action admits a projectively faithful irreducible projective representation on $V^{\otimes n}$.
\end{proposition}
\begin{proof}
    There is an obvious faithful representation $S_n \to U(V^{\otimes n})$ given by permuting the tensor factors. By abuse of notation we will just denote this representation by $\tau \mapsto \tau$, for any $\tau \in S_n$. The function $\pi^{n}_{Sym} \colon G^n \rtimes S_n \to U(V^{\otimes n})$ defined by
    \begin{equation*}
        \pi^{n}_{Sym}(x_1,\dots,x_n,\tau) 
        = (\pi(x_1) \otimes \dots \otimes \pi(x_n)) \, \tau
    \end{equation*}
    is then easily verified to be a projective representation of $G^n \rtimes S_n$ on $V^{\otimes n}$. Furthermore, it is faithful if and only if $\pi$ is faithful, and it is irreducible if $\pi$ is irreducible, completing the proof.
\end{proof}

\section{Stabilizer codes, Clifford codes, and weak stabilizer codes: definitions}\label{sec:codes}

Given an error model $(G,\lambda) \in \EM_V$, the standard way of producing error correcting codes for this model is by using some additional algebraic data, usually in the form of a normal subgroup $N$ of $G$ and an irreducible representation of $N$. The most well-known codes that arise in this way are the stabilizer codes, which were first studied in the concrete case where the error model is the Pauli group, cf. \cite{Gottesman97,Calderbank97}, where the normal subgroup $N$ was assumed to be abelian. Generalizations of this were quickly developed, namely Clifford codes, cf. \cite{KnillI96, KlappeneckerRöttelerI02, KlappeneckerRöttelerII02}, where the error models they considered were nice error bases and any normal subgroup was considered, not just abelian ones.

We will define some codes that generalize these existing codes, and we will study how they are related to each other. All our definitions will be made using projective error models. Note that completely analogous definitions can be made with error models in place of projective error models.

\begin{definition}\label{def:stabilizer code}
    Let $V$ be a Hilbert space and $(G,\pi) \in \PEM_V$. A nonzero subspace $W \subset V$ is a \emph{stabilizer code with respect to the pair $(G,\pi)$} if there exists a normal subgroup $N$ of $G$ and a function $f \colon N \to \bbT$ such that
    \begin{equation*}
        W = V^{St}(G,\pi,N,f) := \{ \xi \in V \colon \pi(x) \xi = f(x)\xi \text{ for all } x \in N \}.
    \end{equation*}
    It will often be clear from context which projective error model we are considering. In this case we will simply refer to $W$ as a \emph{stabilizer code}.
\end{definition}

Note that the requirement of $N$ being a normal subgroup of $G$ is what is usually considered in the literature, see for example \cite{KlappeneckerRöttelerII02}. It is possible to consider the same setup with general subgroups of $G$, and we therefore make the following definition.

\begin{definition}\label{def:weak stabilizer code}
    Let $V$ be a Hilbert space and $(G,\pi) \in \PEM_V$. A nonzero subspace $W \subset V$ is a \emph{weak stabilizer code with respect to the pair $(G,\pi)$} if there exists a subgroup $H$ of $G$ and a function $f \colon H \to \bbT$ such that
    \begin{equation*}
        W = V^{St}(G,\pi,H,f) = \{ \xi \in V \colon \pi(x) \xi = f(x)\xi \text{ for all } x \in H \}.
    \end{equation*}
    We will often only refer to $W$ as a \emph{weak stabilizer code}.
\end{definition}

In \cref{sec:non-Clifford weak stabilizer codes} we present a family of codes that are weak stabilizer codes but not stabilizer codes. Hence, this weaker definition does capture more examples of codes than the definition of a stabilizer code. In fact, even XP-stabilizer codes, cf. \cite{Webster2022}, are only defined in terms of this weak stabilizer code construction, and it is not obvious wether or not they are also stabilizer codes.

It is clear that any stabilizer code is automatically a weak stabilizer code. We can also characterize when the space $V^{St}(G,\pi,H,f)$ is nonzero. First we prove some auxiliary results.

\begin{proposition}\label{prop:H f is sigma projective}
    Let $V$ be a Hilbert space, $(G,\pi) \in \PEM_V$, $H$ a subgroup of $G$, and $f \colon H \to \bbT$ a function. If the subspace $V^{St}(G,\pi,H,f)$ is nonzero, then $f$ is a projective representation of $H$ on $\bbC$. If $\pi$ is $\sigma$-projective, then $f$ is $\Res^G_H\sigma$-projective.
\end{proposition}
\begin{proof}
    Since $V^{St}(G,\pi,H,f)$ is nonzero, there exists a nonzero $\xi \in V^{St}(G,\pi,H,f)$, and for every $x,y \in H$ we have that
    \begin{equation*}
        f(x)f(y)\xi 
        = \pi(x)\pi(y)\xi
        = \sigma(x,y)\pi(xy)\xi
        = \sigma(x,y)f(xy)\xi.
    \end{equation*}
    Hence, $f(x)f(y) = \sigma(x,y)f(xy)$ for all $x,y \in H$, completing the proof.
\end{proof}

\begin{corollary}\label{cor:restriction of sigma to H is trivial}
    Let $V$ be a Hilbert space, $(G,\pi) \in \PEM_V$, $H$ a subgroup of $G$, and $f \colon H \to \bbT$ a function. Suppose that $\pi$ is $\sigma$-projective. If the subspace $V^{St}(G,\pi,H,f)$ is nonzero, then $\Res^G_H\sigma$ is a coboundary.
\end{corollary}
\begin{proof}
    By \cref{prop:H f is sigma projective} we have that $\Res^G_H\sigma = \delta f$, completing the proof.
\end{proof}

\begin{theorem}\label{thm:nonzero Stabilizer codes}
    Let $V$ be a Hilbert space, $(G,\pi) \in \PEM_V$ with $\pi$ being $\sigma$-projective, and $H$ a subgroup of $G$. Consider the following statements:
    \begin{enumerate}
        \item There exists a function $f \colon H \to \bbT$ such that the subspace $V^{St}(G,\pi,H,f)$ is nonzero.
        \item There exists a function $f \colon H \to \bbT$ such that $\Res^G_H\sigma = \delta f$  and the vector space of intertwiners $\Hom_H(f,\Res^G_H\pi)$ is nonzero.
        \item The subgroup $H$ is abelian and $\Res^G_H\sigma$ is a coboundary.
    \end{enumerate}
    We have that $(2)$ is equivalent to $(1)$, and that $(3)$ implies $(2)$. Furthermore, if $H$ is a normal subgroup of $G$, then $(2)$ implies $(3)$.
\end{theorem}
\begin{proof}
    That (1) implies (2) is \cref{prop:H f is sigma projective} and \cref{cor:restriction of sigma to H is trivial}. To show that (2) implies (1) let $T \in \Hom_H(f, \Res^G_H\pi)$ be a nonzero intertwiner. Then $\xi := T(1) \in V$ is nonzero, and since $T$ is an intertwiner we have that
    \begin{equation*}
        \pi(x) \xi 
        = \pi(x) T(1)
        = T (f(x)1)
        = f(x) T(1)
        = f(x) \xi, \text{ for all } x \in H.
    \end{equation*}
    Hence, $\xi \in V^{St}(G,\pi,H,f)$, showing that this space is nonzero.

    To show that (3) implies (2), write $\Res^G_H\sigma = \delta f'$ for some function $f' \colon H \to \bbT$. Then the map $(\overline{f'} \, \Res^G_H\pi) \colon H \to U(V)$ defined by
    \begin{equation*}
        (\overline{f'} \, \Res^G_H)(x) = \overline{f'(x)} \pi(x),
    \end{equation*}
    is a group homomorphism, or rather a representation of $H$ on $V$. Since $H$ is abelian we can write $\overline{f'} \, \Res^G_H\pi$ as a direct sum of 1-dimensional representations. Let $\chi \colon H \to \bbT$ be one of these representations. By choice of $\chi$ we have that $\Hom_H(\chi, \overline{f'} \, \Res^G_H\pi)$ is nonzero, hence the vector space of ($\Res^G_H\sigma$-projective) intertwiners $\Hom_H(f'\chi, \Res^G_H\pi)$ is nonzero. Define $f = f'\chi$. Then we have that
    \begin{equation*}
        \Res^G_H\sigma 
        = \delta f' 
        = (\delta\chi) \, (\delta f')
        = \delta (\chi f')
        = \delta f,
    \end{equation*}
    showing that (3) indeed implies (2).

    Now assume that $H$ is a normal subgroup of $G$. To show that (2) implies (3), note that we assume that the vector space of intertwiners $\Hom_H(f,\Res^G_H\pi)$ is nonzero. Hence, by Mackey's Lemma we have that
    \begin{equation*}
        \Res^G_H\pi \simeq \ell \bigoplus_{r \in G/I_G(f)} f^r,
    \end{equation*}
    where $\ell = \dim \Hom_H(f,\Res^G_H\pi)$. Since all the $\Res^G_H\sigma$-projective representations $f^r$ are 1-dimensional we can make the following computation for each $x,y \in H$,
    \begin{align*}
        \sigma(x,y)\pi(xy) 
        & = \ell \bigoplus_{r \in G/I_G(f)} \sigma(x,y)f^r(xy) \\
        & = \ell \bigoplus_{r \in G/I_G(f)} f^r(x) f^r(y) \\
        & = \ell \bigoplus_{r \in G/I_G(f)} f^r(y) f^r(x) \\
        & = \ell \bigoplus_{r \in G/I_G(f)} \sigma(y,x)f^r(yx) \\
        & = \sigma(y,x)\pi(yx).
    \end{align*}
    Since $\pi$ is a projectively faithful $\sigma$-projective representation of $G$ we get that $xy = yx$ for all $x,y \in H$. Hence, $H$ is abelian, completing the proof.
\end{proof}

\begin{corollary}\label{cor:abelian error models}
    Let $V$ be a Hilbert space. Suppose that $(G,\pi) \in \PEM_V$ is a projective error model with $G$ being abelian and $\pi$ being $\sigma$-projective. Suppose that $H$ is a subgroup of $G$. Then there exists a function $f \colon H \to \bbT$ such that $V^{St}(G,\pi,H,f)$ is nonzero if and only if $\Res^G_H\sigma$ is a coboundary.
\end{corollary}
\begin{proof}
    Since any subgroup of $G$ is automatically normal and abelian, we have that there exists a function $f \colon H \to \bbT$ such that $V^{St}(G,\pi,H,f)$ is nonzero if and only if $\Res^G_H\sigma$ is a coboundary by \cref{thm:nonzero Stabilizer codes}.
\end{proof}

We are furthermore able to determine the dimension of a weak stabilizer code completely in terms of representation theoretic data.

\begin{proposition}\label{prop:dim of weak stabilizer code}
    Let $V$ be a Hilbert space and $(G,\pi) \in \PEM_V$. If $H$ is a subgroup of $G$ and $f \colon H \to \bbT$ is a function such that $V^{St}(G,\pi,H,f)$ is nonzero, then
    \begin{equation*}
        \dim V^{St}(G,\pi,H,f) = \dim \Hom_H(f, \Res^G_H \pi).
    \end{equation*}
    In particular, 
    \begin{equation*}
        \dim V^{St}(G,\pi,H,f) = \frac{1}{|H|} \sum_{x \in H} \overline{f(x)} \chi_\pi(x).
    \end{equation*}
\end{proposition}
\begin{proof}
    By \cref{prop:H f is sigma projective} $f$ is a projective representation of $H$, hence we may consider the vector space of intertwiners $\Hom_H(f, \Res^G_H \pi)$, and by \cref{thm:nonzero Stabilizer codes}, this space is nonzero. We claim that the map $\Hom_H(f, \Res^G_H \pi) \to V^{St}(G,\pi,H,f)$ defined by $T \mapsto T(1)$ is a linear isomorphism. 
    
    It is clear that this map is linear and injective, and by the same argument as in the proof of \cref{thm:nonzero Stabilizer codes}, we have that $T(1)$ is indeed an element of $V^{St}(G,\pi,H,f)$. To see that it is surjective, suppose that $\xi \in V^{St}(G,\pi,H,f)$. Define $T_\xi \colon \bbC \to V^{St}(G,\pi,H,f)$ by
    \begin{equation*}
        T_\xi(z) = z \xi, \text{ for } z \in \bbC.
    \end{equation*}
    Then for any $x \in H$,
    \begin{equation*}
        T_\xi(f(x))
        = f(x) \xi
        = \pi(x) \xi
        = \pi(x) T_\xi(1).
    \end{equation*}
    Hence, $T_\xi \in \Hom_H(f, \Res^G_H \pi)$, which shows that
    \begin{equation*}
        \dim V^{St}(G,\pi,H,f) = \dim \Hom_H(f, \Res^G_H \pi).
    \end{equation*}
    By \cref{eqn:inner porduct of characters} and \cref{eqn:inner product} we get that this is furthermore equal to
    \begin{equation*}
        \frac{1}{|H|} \sum_{x \in H}f(x)\overline{\chi_\pi(x)}.
    \end{equation*}
\end{proof}

Note that stabilizer codes and weak stabilizer codes give a way to construct the largest code space that is stabilized by a subgroup of $G$. Clifford codes, which we define next, are qualitatively different as they are constructed from a subgroup of $G$ that will correspond to the logical operations on the code.

\begin{definition}\label{def:Clifford code}
    Let $V$ be a Hilbert space and $(G,\pi) \in \PEM_V$. A subspace $W \subset V$ is a \emph{Clifford code with respect to the pair $(G,\pi)$} if there exists a subgroup $L$ of $G$ such that $W$ is:
    \begin{itemize}
        \item $(\Res^G_L\pi)$-.invariant,
        \item the projective representation $(\Res^G_L\pi)|_W$ is irreducible, and
        \item $\Ind^G_L((\Res^G_L\pi)|_W) \simeq \pi$.
    \end{itemize}
    We will often refer to $W$ as a \emph{Clifford code}.
\end{definition}

\begin{proposition}\label{prop:abstract Clifford code}
    Let $V$ be a Hilbert space and $(G,\pi) \in \PEM_V$. If $L$ is a subgroup of $G$ and $\rho \colon L \to U(V_\rho)$ is an irreducible projective representation such that $\Ind^G_L\rho \simeq \pi$, then there exists a unique subspace $V^{Cl}(G,\pi,L,\rho)$ of $V$ with the property that $\Res^G_L\pi$ restricted to $V^{Cl}(G,\pi,L,\rho)$ is an irreducible projective representation of $L$ and $(\Res^G_L\pi)|_{V^{Cl}(G,\pi,L,\rho)} \simeq \rho$. 
\end{proposition}
\begin{proof}
    Let $T \in \Hom_L(\rho, \Res^G_L\pi)$. Since $\ker T$ is an $L$-invariant subspace of $V_\rho$, we have that $T$ is either injective or the $0$ intertwiner. By Frobenius reciprocity and Schur's Lemma we have that
    \begin{equation*}
        \dim\Hom_L(\rho,\Res^G_L\pi) = \dim\Hom_G(\Ind^G_L\rho,\pi) = 1.
    \end{equation*}
    Hence, there exists a nonzero intertwiner $T \colon V_\rho \to V$, and this intertwiner is unique up to scaling. Define $V^{Cl}(G,\pi,L,\rho) := T(V_\rho)$. Since $T$ is unique up to scaling, we have that $T(V_\rho) = T'(V_\rho)$ for any other nonzero intertwiner $T'$, making the subspace $V^{Cl}(G,\pi,L,\rho)$ independent of choice of intertwiner.

    Lastly, it is clear that $\rho \simeq (\Res^G_L\pi)|_{V^{Cl}(G,\pi,L,\rho)}$, namely $T \colon V_\rho \to V^{Cl}(G,\pi,L,\rho)$ is an intertwiner displaying this isomorphism. This completes the proof.
\end{proof}

\begin{proposition}\label{prop:abstract Clifford code is Clifford code}
    Let $V$ be a Hilbert space and $(G,\pi) \in \PEM_V$. A subspace $W \subset V$ is a Clifford code if and only if there exist a subgroup $L$ of $G$ and an irreducible projective representation $\rho \colon L \to U(V_\rho)$ with $\Ind_L^G\rho \simeq \pi$ and $W = V^{Cl}(G,\pi,L,\rho)$. 
\end{proposition}
\begin{proof}
    It is clear from \cref{prop:abstract Clifford code} that if $W = V^{Cl}(G,\pi,L,\rho)$, then $W$ is a Clifford code as in \cref{def:Clifford code}. On the other hand, if $W$ is a Clifford code, then there exists a subgroup $L$ of $G$ such that $W$ is $(\Res^G_L\pi)$-invariant and irreducible, with $\Ind^G_L((\Res^G_L\pi)|_W) \simeq \pi$. We would like to show that the inclusion map $i \colon W \to V$ is an intertwiner in $\Hom_L((\Res^G_L\pi)|_W,\Res^G_L\pi)$. In that case,
    \begin{equation*}
        W = i(W) = V^{Cl}(G,\pi,L,(\Res^G_L\pi)|_W)
    \end{equation*}
    by the definition of the space $V^{Cl}(G,\pi,L,(\Res^G_L\pi)|_W)$. But $W$ is assumed to be $(\Res^G_L\pi)$-invariant, hence the inclusion is an intertwiner, completing the proof.
\end{proof}

\begin{remark}\label{rem:more general notion of Clifford code}
    Note that our definition of a Clifford code is slightly more general than that found in \cite{KnillI96, KlappeneckerRöttelerI02, KlappeneckerRöttelerII02}. Their initial data to construct the space $V^{Cl}(G,\pi,L,\rho)$ is a normal subgroup $N$ of $G$ and an irreducible projective representation $\theta$ of $N$ such that the vector space of intertwiners $\Hom_N(\theta, \Res^G_N\pi)$ is nonzero. The subgroup $L$ of $G$ is then chosen to be the inertia group of $\theta$, $I_G(\theta)$. Then they require that $\rho$ is the unique irreducible projective representation of $I_G(\theta)$ such that $\Hom_N(\theta, \Res^L_N\rho)$ is nonzero and $\Ind^G_L\rho \simeq \pi$. Such a projective representation is guaranteed to exists by Clifford correspondence.

    It is not clear that our definition of a Clifford code truly is more general than any of the definitions made prior to this work. The benefit of our definition is however that it only contains the technicalities that are truly needed to make powerful arguments about Clifford codes. This has the effect of simplifying some arguments. In reality, naturally stumbling upon subgroups $L$ that satisfy all the assumptions in \cref{def:Clifford code} is difficult unless $L$ is the inertia group of some irreducible representation.
\end{remark}

To end this section we show that every stabilizer code is a Clifford code.

\begin{proposition}\label{prop:stabilizer codes are Clifford codes}
    Let $V$ be a Hilbert space, $(G,\pi) \in \PEM_V$, $N$ be a normal subgroup of $G$, and $f \colon N \to \bbT$ be a function. If $W = V^{St}(G,\pi,N,f)$ is a stabilizer code, then $W$ is a Clifford code.
\end{proposition}
\begin{proof}
    Since $W$ is a stabilizer code, we have by definition that $V^{St}(G,\pi,N,f)$ is nonzero. Thus, by the equivalence of (1) and (2) in \cref{thm:nonzero Stabilizer codes} we have that the vector space of intertwiners $\Hom_H(f, \Res^G_N\pi)$ is nonzero. Let $L = I_G(f)$. Then by Clifford correspondence there exist a unique irreducible projective representation $\rho \colon L \to U(V_\rho)$ such that $\Ind^G_L\rho \simeq \pi$ and the vector space of intertwiners $\Hom_N(f,\Res^L_N\rho)$ is nonzero. Let $W' = V^{Cl}(G,\pi,L,\rho)$. By \cref{prop:abstract Clifford code is Clifford code} we have that
    \begin{equation*}
        \rho \simeq (\Res^G_L\pi)|_{W'}.
    \end{equation*}
    Combining this with Mackey's Lemma we have that
    \begin{equation*}
        (\Res^G_N\pi)|_{W'} \simeq \Res^L_N\rho \simeq \ell f,
    \end{equation*}
    where $\ell = \dim\Hom_N(f,\Res^L_N\rho)$. Hence, for every $\xi \in W'$ we obtain
    \begin{equation*}
        \pi(x) \xi = f(x) \xi, \text{ for all } x \in N.
    \end{equation*}
    Hence, $W' \subset W$. By Clifford correspondence we also have that $\dim \Hom_N(f,\Res^G_N\pi) = \ell$, hence $\dim W' = \dim W$. Thus, $W' = W$ completing the proof.
\end{proof}

To summarize, we have defined three different types of codes: stabilizer codes, weak stabilizer codes, and Clifford codes. They are related in the following heuristic:
\begin{equation*}
    \{\text{Clifford codes}\}
    \supset \{\text{Stabilizer codes}\}
    \subset \{\text{Weak stabilizer codes}\}.
\end{equation*}
These inclusions are strict in general, which we will see in \cref{sec:non-stabilizer Clifford codes} and \cref{sec:non-Clifford weak stabilizer codes}. Some interesting questions are
\begin{enumerate}
    \item Do we have
    \begin{equation*}
        \{\text{Stabilizer codes}\} = \{\text{Clifford codes}\} \cap \{\text{Weak stabilizer codes}\}?
    \end{equation*}
    \item Are there good criteria to determine when a Clifford code is a stabilizer code?
    \item And similarly, are there good criteria to determine when a weak stabilizer code is a stabilizer code?
\end{enumerate}
Questions (1) and (3) appear to be new. Question (2) is answered in \cite{KlappeneckerRötteler04} and \cite{GrasslMartínezNicolás10} for groups of central type, we will get back to this question in \cref{sec:weak stabilizer Clifford codes}.

\section{Clifford codes: error detection}\label{sec:error correction}

In this section we want to study the collection of detectable errors for Clifford codes. By \cref{prop:stabilizer codes are Clifford codes} this will also give a description of the detectable errors for a stabilizer code. We first make a couple of definitions to establish some notation.

\begin{definition}\label{def:invariants and stabilizers}
    Let $V$ be a Hilbert space, $(G,\pi) \in \PEM_V$, and $W \subset V$ be a subspace. We define the sets
    \begin{align*}
        L_{(G,\pi)}(W) & := \{ x \in G \colon P_W \pi(x) = \pi(x) P_W \}, \\
        S_{(G,\pi)}(W) & := \{ x \in G \colon P_W \pi(x) P_W \in \bbT P_W \}, 
    \end{align*}
    where $P_W$ denotes the orthogonal projection onto $W$. One can of course make analogous definitions for error models in place of projective error models.
\end{definition}

Notice that both $L_{(G,\pi)}(W)$ and $S_{(G,\pi)}(W)$ are subgroups of $G$, and in fact $S_{(G,\pi)}(W)$ is a subgroup of $L_{(G,\pi)}(W)$. The group $L_{(G,\pi)}(W)$ is the group of logical operations on the code $W$, and the group $S_{(G,\pi)}(W)$ is the group of stabilizers.

We now want to consider a class of codes that are nice enough so that we may describe the set $D_{(G,\pi)}(W)$ of detectable errors in terms of these two groups.

\begin{proposition}\label{prop:detectable errors in terms of inertia group and stabilizers}
    Let $V$ be a Hilbert space, $(G,\pi) \in \PEM_V$, and $W \subset V$ be a subspace. Then the following conditions are equivalent:
    \begin{equation}\label{eqn:action on code space}
        \pi(x) W \subset W \text{ or } \pi(x) W \subset W^\perp \text{ for every } x \in G.
    \end{equation}
    \begin{equation}\label{eqn:detectable errors}
        D_{(G,\pi)}(W) = (G \setminus L_{(G,\pi)}(W)) \cup S_{(G,\pi)}(W).
    \end{equation}
\end{proposition}
\begin{proof}
    Suppose that $\pi(x) W \subset W$ or $\pi(x) W \subset W^\perp$ for every $x \in G$. If $x \in G$ is such that $\pi(x) W \subset W$, then $x \in L_{(G,\pi)}(W)$, hence $x \in D_{(G,\pi)}(W)$ if and only if $x \in S_{(G,\pi)}(W)$. Furthermore, if $x \in G \setminus L_{(G,\pi)}(W)$ then $\pi(x) W \subset W^\perp$, and it is clear that these are precisely the elements of $D_{(G,\pi)}(W)$ that satisfy $P_W\pi(x)P_W = 0$.

    For the converse, it is clear that if $x \in L_{(G,\pi)}(W)$, then $\pi(x)W \subset W$. If $x \in G \setminus L_{(G,\pi)}(W)$, then $x \in D_{(G,\pi)}(W)$ and satisfies $P_W \pi(x) P_W = 0$. From this it is clear that $\pi(x) W \subset W^\perp$, completing the proof.
\end{proof}
This result should really be thought of as giving a condition ensuring that the detectable errors of a code can be described in terms of the stabilizers and logical operations on the code.

Having \cref{eqn:detectable errors} in mind, we get a new way of thinking about stabilizer codes and Clifford codes. A Clifford code is a code $W$ where one has prescribed the group $L_{(G,\pi)}(W)$ of logical operations, while a stabilizer code is a code $W$ where one has prescribed a subgroup of stabilizers $S_{(G,\pi)}(W)$. We make this precise in the following Theorem, which recovers \cite[Theorem 1]{KlappeneckerRöttelerII02} adapted to the present setting.

\begin{theorem}\label{thm:detectable errors:Clifford code}
    Let $V$ be a Hilbert space and $(G,\pi) \in \PEM_V$. Suppose that $W = V^{Cl}(G,\pi,L,\rho)$ is a Clifford code. Then
    \begin{align*}
        L_{(G,\pi)}(W) & = L, \\
        S_{(G,\pi)}(W) & = \ker(q \circ \rho), \\
        D_{(G,\pi)}(W) & = (G \setminus L) \cup \ker(q \circ \rho).
    \end{align*}
\end{theorem}
\begin{proof}
    Note first that by the definition of a Clifford code, we have that $L \subset L_{(G,\pi)}(W)$. To see that they are equal we will show that for any $x \in G \setminus L$, $P_W \pi(x) P_W = 0$.

    Let $T_1 \colon V \to \Ind^G_L V_\rho$ be an intertwiner exhibiting the isomorphism $\pi \simeq \Ind^G_L\rho$, and similarly, let $T_2 \colon W \to V_\rho$ be an intertwiner exhibiting the isomorphism $\Res^G_L\pi|_W \simeq \rho$. Then we have a commuting square of intertwiners:
    \begin{equation*}
        \xymatrix{
            W \ar[d]_-{T_2} \ar[r]^-{i} & V \ar[d]^-{T_1} \\
            V_\rho \ar[r]_-{j} & \Ind^G_L V_\rho
        }
    \end{equation*}
    Here $i \colon W \to V$ is the inclusion, and $j \colon V_\rho \to \Ind^G_L V_\rho$ is the map defined by $j(\eta) = 1_G \otimes \eta$. Let $x \in G$ and write $x = ry$ with $y \in L$ and $r \in G \setminus L$. If $\xi \in W$, then we compute that
    \begin{align*}
        P_W \pi(x) i(\xi)
        & = P_W \pi(x) T_1^{-1} T_1(i(\xi)) \\
        & = P_W T_1^{-1} \Ind^G_L\rho(x) T_1(i(\xi)) \\
        & = P_W T_1^{-1} \Ind^G_L\rho(x) j(T_2(\xi)) \\
        & = P_W T_1^{-1} \Ind^G_L\rho(x) (1_G \otimes T_2(\xi)) \\
        & = P_W T_1^{-1} (x \otimes T_2(\xi)) \\
        & = P_W T_1^{-1} (ry \otimes T_2(\xi)) \\
        & = P_W T_1^{-1} (r \otimes \rho(y)T_2(\xi)).
    \end{align*}
    Notice that $P_W T_1^{-1} (r \otimes \rho(y)T_2(\xi))$ is zero if and only if $r \otimes \rho(y)T_2(\xi)$ is not of the form $j(\eta)$ for any $\eta \in V_\rho$. This happens if and only if $r \neq 1_G$, which is true if and only if $x \in G \setminus L$. Hence, $P_W \pi(x) P_W = 0$ whenever $x \in G \setminus L$ proving that
    \begin{equation*}
        L_{(G,\pi)}(W) = L.
    \end{equation*}

    Furthermore, this shows that $\pi(x)W \subset W$ for every $x \in L$ and that $\pi(x)W \subset W^\perp$ for every $x \in G \setminus L$. Hence, by \cref{prop:detectable errors in terms of inertia group and stabilizers} we get that
    \begin{equation*}
        D_{(G,\pi)}(W) = (G \setminus L_{(G,\pi)}(W)) \cup S_{(G,\pi)}(W).
    \end{equation*}
    It remains to show that $S_{(G,\pi)}(W) = \ker(q \circ \rho)$.

    Now suppose that $P_W T_1^{-1} (r \otimes \rho(y)T_2(\xi))$ is a nonzero scalar multiple of $\xi$. By the above argument we know that $x \in L$, hence
    \begin{equation*}
        P_W\pi(x) i(\xi)
        = P_W i T_2^{-1}(\rho(x)T_2(\xi)).
    \end{equation*}
    This is a nonzero scalar multiple of $\xi$ if and only if $\rho(x)$ is a scalar multiple of $1_{U(V_\rho)}$, i.e., if and only if $x \in \ker(q \circ \rho)$. This shows that
    \begin{equation*}
        S_{(G,\pi)}(W) = \ker(q \circ \rho),
    \end{equation*}
    completing the proof.
\end{proof}

\begin{corollary}\label{cor:Clifford codes has unique inertia group}
    Let $V$ be a Hilbert space and $(G,\pi) \in \PEM_V$. If $W$ is a Clifford code with
    \begin{equation*}
        V^{Cl}(G,\pi,L_1,\rho_1) = W = V^{Cl}(G,\pi,L_2,\rho_2),
    \end{equation*}
    then $L_1 = L_2$ and $\rho_1 \simeq \rho_2$.
\end{corollary}
\begin{proof}
    By \cref{thm:detectable errors:Clifford code} we get that $L_1 = L_{(G,\pi)}(W) = L_2$. Furthermore, we also get that $\rho_1 \simeq (\Res^G_{L_1}\pi)|_W \simeq \rho_2$, completing the proof.
\end{proof}

As for weak stabilizer codes we are also able to determine the dimension of a Clifford code in terms of algebraic data.

\begin{corollary}\label{cor:order of inertia group for Clifford codes}
    Let $V$ be a Hilbert space and $(G,\pi) \in \PEM_V$. If $W$ is a Clifford code then we have that
    \begin{equation*}
        |L_{(G,\pi)}(W)| = \frac{\dim W}{\dim V} |G|.
    \end{equation*}
\end{corollary}
\begin{proof}
    By \cref{thm:detectable errors:Clifford code} we have that $\Ind^G_{L_{(G,\pi)}(W)}((\Res^G_{L_{(G,\pi)}(W)}\pi)|_W) \simeq \pi$. This implies that
    \begin{equation*}
        \dim V 
        = \dim(\Ind^G_{L_{(G,\pi)}(W)} W) 
        = \frac{|G|}{|L_{(G,\pi)}(W)|} \dim W.
    \end{equation*}
    Rearranging this equation completes the proof.
\end{proof}

With the subgroup $S_{(G,\pi)}(W)$ at hand we can state some more properties of weak stabilizer codes that will turn out to be quite useful.

\begin{lemma}\label{lemma:stabilizers of stabilizer code}
    Let $V$ be a Hilbert space and $(G,\pi) \in \PEM_V$. If $W = V^{St}(G,\pi,H,f)$ is a weak stabilizer code, then $H \subset S_{(G,\pi)}(W)$.
\end{lemma}
\begin{proof}
    For every $x \in H$ we have that $P_W\pi(x)P_W = f(x)P_W$, hence $x \in S_{(G,\pi)}(W)$.
\end{proof}

\begin{proposition}\label{prop:stabilizer codes of whole stabilizer group}
    Let $V$ be a Hilbert space and $(G,\pi) \in \PEM_V$. If $W = V^{St}(G,\pi,H,f)$ is a weak stabilizer code, then there exists a function $\tilde{f} \colon S_{(G,\pi)}(W) \to \bbT$ such that $W = V^{St}(G,\pi,S_{(G,\pi)}(W),\tilde{f})$.
\end{proposition}
\begin{proof}
    For $x \in S_{(G,\pi)}(W)$ define $\tilde{f}(x)$ to be the unique scalar such that $P_W\pi(x)P_W = \tilde{f}(x)P_W$. By \cref{lemma:stabilizers of stabilizer code} we have that $H \subset S_{(G,\pi)}(W)$, hence
    \begin{equation*}
        V^{St}(G,\pi,S_{(G,\pi)}(W),\tilde{f}) \subset V^{St}(G,\pi,H,f) = W.
    \end{equation*}
    On the other hand, by the definition of $S_{(G,\pi)}(W)$, we have that
    \begin{equation*}
        W \subset V^{St}(G,\pi,S_{(G,\pi)}(W),\tilde{f}),
    \end{equation*}
    hence they are equal, completing the proof.
\end{proof}

\cref{thm:detectable errors:Clifford code} can be used to determine if a weak stabilizer code is also a Clifford code, as this result gives strict criteria on the form of the groups $L_{(G,\pi)}(W)$ and $S_{(G,\pi)}(W)$. Similarly, \cref{prop:stabilizer codes of whole stabilizer group} is helpful to determine if a Clifford code is also a weak stabilizer code, as a Clifford code $W = V^{Cl}(G,\pi,L,\rho)$ need not be equal to $V^{St}(G,\pi,S_{(G,\pi)}(W),\tilde{f})$.

Furthermore, \cref{lemma:stabilizers of stabilizer code} gives an answer to the second of our original questions: given some noise we want to protect against, what is the largest code space that can detect this given noise? If $W = V^{St}(G,\pi,H,f)$ is a weak stabilizer code, then $H \subset S_{(G,\pi)}(W) \subset D_{(G,\pi)}(W)$, hence the code space $W$ can detect all the noise that is being modeled by the subgroup $H$ of $G$. It is also clear from the definition of a weak stabilizer code that $W$ is the largest code space that can detect this noise.

\section{Weak stabilizer Clifford codes}\label{sec:weak stabilizer Clifford codes}

In this section we give a complete characterization of when a Clifford code is also a weak stabilizer code for projective error models $(G,\pi) \in \PEM_V$ where $|G| = (\dim V)^2$, i.e., $G$ is a group of central type. This generalizes \cite[Theorem 3]{KlappeneckerRötteler04} and \cite[Corollary 4.2]{GrasslMartínezNicolás10}. We need the following elementary fact about the characters of groups of central type. This fact can be established from results found within \cite{CeccheriniSilbersteinTullio22} and \cite{Cheng15}, for the ease of the reader we choose to present a proof.

\begin{lemma}\label{lemma:character of group of central type}
    Let $G$ be a group of central type, and $\pi \colon G \to U(V)$ be an irreducible projective representation such that $|G| = (\dim V)^2$. Then we have that $\chi_{\pi}(x) = 0$ if and only if $x \neq 1_G$.
\end{lemma}
\begin{proof}
    By Schur's Lemma and \cref{eqn:inner porduct of characters} we have that
    \begin{align*}
        1
        & = \dim \Hom_G(\pi, \pi) \\
        & = \langle \chi_\pi, \chi_\pi \rangle \\
        & = \frac{1}{|G|}\sum_{x \in G}\overline{\chi_\pi(x)}\chi_\pi(x) \\
        & = \frac{\overline{\chi_\pi(1_G)}\chi_\pi(1_G)}{|G|} + \frac{1}{|G|}\sum_{x \in G \setminus \{1_G\}}\overline{\chi_\pi(x)}\chi_\pi(x) \\
        & = \frac{(\dim V)^2}{|G|} + \frac{1}{|G|}\sum_{x \in G \setminus \{1_G\}}\overline{\chi_\pi(x)}\chi_\pi(x) \\
        & = 1 + \frac{1}{|G|}\sum_{x \in G \setminus \{1_G\}}\overline{\chi_\pi(x)}\chi_\pi(x).
    \end{align*}
    Hence, $\chi_\pi(x) = 0$ for all $x \in G \setminus \{1_G\}$, completing the proof.
\end{proof}

\begin{lemma}\label{lemma:unique stabilizer of weak stabilizer Clifford code}
    Let $V$ be a Hilbert space and $(G,\pi) \in \PEM_V$ with $|G| = (\dim V)^2$. If $W \subset V$ is a weak stabilizer code, and $W = V^{St}(G,\pi,H,f)$ for some function $f \colon H \to \bbT$, then $H = S_{(G,\pi)}(W)$, and
    \begin{equation*}
        \dim W = \frac{1}{|S_{(G,\pi)}(W)|} \dim V.
    \end{equation*}
\end{lemma}
\begin{proof}
    By \cref{prop:dim of weak stabilizer code}, and \cref{lemma:character of group of central type} we have that
    \begin{equation*}
        \dim W 
        = \frac{1}{|H|}\sum_{x \in H} \overline{f(x)}\chi_{\pi}(x)
        = \frac{1}{|H|} \dim V.
    \end{equation*}
    By \cref{lemma:stabilizers of stabilizer code} we have that $H \subset S_{(G,\pi)}(W)$, and by \cref{prop:stabilizer codes of whole stabilizer group} there exists a function $\tilde{f} \colon S_{(G,\pi)}(W) \to \bbT$ such that $W = V^{St}(G,\pi,S_{(G,\pi)}(W),\tilde{f})$. Hence, we can also compute that
    \begin{equation*}
        \dim W = \frac{1}{|S_{(G,\pi)}(W)|} \dim V.
    \end{equation*}
    Thus, $H = S_{(G,\pi)}(W)$, completing the proof.
\end{proof}

\begin{theorem}\label{thm:weak stabilizer Clifford code}
    Let $V$ be a Hilbert space and $(G,\pi) \in \PEM_V$ with $|G| = (\dim V)^2$. If $W \subset V$ is a Clifford code, then $W$ is a weak stabilizer code if and only if $|G| = |L_{(G,\pi)}(W)|\cdot|S_{(G,\pi)}(W)|$.
\end{theorem}
\begin{proof}
    Suppose that $W$ is a weak stabilizer code. Let $H$ be a subgroup of $G$ and $f \colon H \to \bbT$ be a function such that $W = V^{St}(G,\pi,H,f)$. By \cref{cor:order of inertia group for Clifford codes} we have that
    \begin{equation*}
        \dim W = \frac{|L_{(G,\pi)}(W)|}{|G|} \dim V.
    \end{equation*}
    By \cref{lemma:unique stabilizer of weak stabilizer Clifford code} we have that
    \begin{equation*}
        \dim W = \frac{1}{|S_{(G,\pi)}|} \dim V.
    \end{equation*}
    By comparing these two equations for $\dim W$ we get that $|G| = |L_{(G,\pi)}(W)|\cdot|S_{(G,\pi)}|$.

    To prove the converse statement, suppose that $|G| = |L_{(G,\pi)}(W)|\cdot|S_{(G,\pi)}(W)|$. By definition of $S_{(G,\pi)}(W)$ there exist a function $f \colon S_{(G,\pi)}(W) \to \bbT$ such that
    \begin{equation*}
        W \subset W' := V^{St}(G,\pi,S_{(G,\pi)}(W),f).
    \end{equation*}
    By \cref{cor:order of inertia group for Clifford codes} and \cref{lemma:unique stabilizer of weak stabilizer Clifford code}
    \begin{equation*}
        \dim W
        = \frac{|L_{(G,\pi)}(W)|}{|G|} \dim V
        = \frac{1}{|S_{(G,\pi)}(W)|} \dim V
        = \dim W'.
    \end{equation*}
    Hence, $W = W'$, completing the proof.
\end{proof}

We choose to highlight the following fact which was shown in the proof above.

\begin{proposition}\label{prop:failure of being a weak stabilizer code}
    Let $V$ be a Hilbert space and $(G,\pi) \in \PEM_V$ with $|G| = (\dim V)^2$. If $W \subset V$ is a non-weak stabilizer Clifford code, then
    \begin{equation*}
        |L_{(G,\pi)}(W)| \cdot |S_{(G,\pi)}(W)| < |G|.
    \end{equation*}
\end{proposition}

From \cref{thm:weak stabilizer Clifford code} we immediately get the following result.

\begin{corollary}\label{cor:stabilizer Clifford code}
    Let $V$ be a Hilbert space and $(G,\pi) \in \PEM_V$ with $|G| = (\dim V)^2$. If $W \subset V$ is a Clifford code, then $W$ is a stabilizer code if and only if $S_{(G,\pi)}(W)$ is a normal subgroup of $G$ and $|G| = |L_{(G,\pi)}(W)|\cdot|S_{(G,\pi)}(W)|$.
\end{corollary}
\begin{proof}
    By \cref{thm:weak stabilizer Clifford code} we have that $W$ is a weak stabilizer code if and only if
    \begin{equation*}
        |G| = |L_{(G,\pi)}(W)|\cdot|S_{(G,\pi)}(W)|.
    \end{equation*}
    In this case, we have by \cref{lemma:unique stabilizer of weak stabilizer Clifford code} that $S_{(G,\pi)}(W)$ is the only subgroup of $G$ for which there exists a function $f \colon S_{(G,\pi)}(W) \to \bbT$ such that
    \begin{equation*}
        W = V^{St}(G,\pi,S_{(G,\pi)}(W),f).
    \end{equation*}
    Hence, $W$ is a stabilizer code if and only if $S_{(G,\pi)}(W)$ is a normal subgroup of $G$. This completes the proof.
\end{proof}

\section{Examples of non-stabilizer Clifford codes}\label{sec:non-stabilizer Clifford codes}

In \cite[Section 10.9]{KlappeneckerRötteler02} the first and smallest example of a non-stabilizer Clifford code is constructed. Since the authors mainly were concerned with projective error models from groups of central type, or nice error bases, it is not mentioned that this example fits into an infinite family of similar examples. Inserting $n = 2$ in the next example retrieves the well known example by Klappenecker and R{\"o}tteler.

\begin{proposition}\label{prop:smallest non-stabilizer Clifford code}
    Let $n \geq 2$ be an integer and consider the group $C_2 \times D_{2n}$ with the following presentation:
    \begin{equation*}
        C_2 \times D_{2n} = \langle a, b, c \, \vert \, a^{2n} = b^2 = c^2 = [a,c] = [b,c] = 1, \, bab = a^{-1} \rangle.
    \end{equation*}
    Let $X$ and $P$ be the matrices as in \cref{prop:projective XP-error model}, and $\rho \colon D_{2n} \to U(\bbC^2)$ be the corresponding projective representation. Define $\pi \colon C_2 \times D_{2n} \to U(\bbC^4)$ by
    \begin{equation*}
        \pi(c^kb^la^m) =
        \begin{bmatrix}
            0 & I_2 \\
            I_2 & 0
        \end{bmatrix}^k
        \begin{bmatrix}
            X & 0 \\
            0 & X
        \end{bmatrix}^l
        \begin{bmatrix}
            P & 0 \\
            0 & -P
        \end{bmatrix}^m.
    \end{equation*}
    Then $\pi$ is a projectively faithful irreducible projective representation of $C_2 \times D_{2n}$ on $\bbC^4$. Moreover, $\Ind_{D_{2n}}^{C_2 \times D_{2n}} \rho \simeq \pi$, and $W := V^{Cl}(C_2 \times D_{2n}, \pi, D_{2n}, \rho)$ is a non-stabilizer Clifford code.
\end{proposition}
\begin{proof}
    As in the proof of \cref{prop:projective XP-error model} one computes that $\pi$ is indeed a projective representation of $C_2 \times D_{2n}$ on $\bbC^4$. To see that $\pi$ is irreducible one confirms that
    \begin{equation*}
        \Span\{\pi(x) \colon x \in C_2 \times D_{2n}\} = M_4(\bbC).
    \end{equation*}
    To see that $\pi$ is projectively faithful, notice that $\pi(c^kb^la^m)$ is a diagonal matrix if and only if $k = l = 0$. In this case it is easy to see that $\pi(c^kb^la^m)$ is a scalar multiple of the identity if and only if $m = 0$ as well, hence $\pi$ is projectively faithful.

    To see that $\Ind_{D_{2n}}^{C_2 \times D_{2n}}\rho \simeq \pi$ define the linear map $T \colon \Ind_{D_{2n}}^{C_2 \times D_{2n}}\bbC^2 \to \bbC^4$, by
    \begin{equation*}
        T(x \otimes (z_1,z_2)) = \pi(x)(z_1,z_2,0,0).
    \end{equation*}
    This is an intertwiner, and we notice that the Clifford code $W$ corresponds to the first two coordinates in $\bbC^4$.

    By \cref{prop:projective XP-error model} we know that $\rho$ is projectively faithful. Hence, by \cref{thm:detectable errors:Clifford code} we get that
    \begin{equation*}
        S_{(C_2 \times D_{2n}, \pi)}(W) = \ker(q \circ \rho) = \{1\}.
    \end{equation*}
    Then
    \begin{equation*}
        V^{St}(C_2 \times D_{2n}, \pi,S_{(C_2 \times D_{2n}, \pi)}(W),1) = \bbC^4.
    \end{equation*}
    By \cref{prop:stabilizer codes of whole stabilizer group} this implies that $W$ is not a weak stabilizer code, completing the proof.
\end{proof}

The next family of examples is inspired by the above example by Klappenecker and R{\"o}tteler, as well as the work by Ginosar and Schnabel in \cite[Theorem A]{GinosarSchnabel19}, where they present a large family of groups of central type.

\begin{proposition}\label{prop:family of non-stabilizer Clifford codes}
    Let $n \geq 3$ be an odd number. Let $\bbZ_2$ act on $\bbZ_n \times \bbZ_n$ via inversion, and let $L := (\bbZ_n \times \bbZ_n) \rtimes \bbZ_2$ denote the associated semi-direct product. Define $\rho \colon L \to U(\bbC^n)$ by
    \begin{equation*}
        \rho(a,b,c) = X_n^a Z_n^b C^c,
    \end{equation*}
    where $X_n$ and $Z_n$ are as in \cref{prop:projective error models exist} and
    \begin{equation*}
        C :=
        \begin{bmatrix}
            0 & 0 & \cdots & 0 & 1 \\
            0 & 0 & \cdots & 1 & 0 \\
            \vdots & \vdots & \ddots & \vdots & \vdots \\
            0 & 1 & \cdots & 0 & 0 \\
            1 & 0 & \cdots & 0 & 0
        \end{bmatrix}.
    \end{equation*}
    Set $G := L \times \bbZ_2$ and define $\pi \colon G \to U(\bbC^{2n})$ by
    \begin{equation*}
        \pi(a,b,c,d) =
        \begin{bmatrix}
            X_n & 0 \\
            0 & X_n 
        \end{bmatrix}^a
        \begin{bmatrix}
            Z_n & 0 \\
            0 & Z_n 
        \end{bmatrix}^b
        \begin{bmatrix}
            C & 0 \\
            0 & -C 
        \end{bmatrix}^c
        \begin{bmatrix}
            0 & I_n \\
            I_n & 0 
        \end{bmatrix}^d.
    \end{equation*}
    Then both $\rho$ and $\pi$ are projectively faithful irreducible projective representations, and $\Ind_L^G\rho \simeq \pi$. Furthermore, $W := V^{Cl}(G,\pi,L,\rho)$ is a non-stabilizer Clifford code.
\end{proposition}
\begin{proof}
    The fact that $\rho$ is a projectively faithful irreducible projective representation of $L$ follows by \cref{prop:projective error models exist}. To see that $\pi$ is a projectively faithful irreducible projective representations of $G$, we claim that the collection $\{ \pi(x) : x \in G \}$ forms a basis for $M_{2n}(\bbC)$. Indeed, if $\pi(x)$ and $\pi(y)$ are linearly dependent, then we get that $\pi(x^{-1}y)$ is a scalar multiple of the identity matrix. Notice that $\pi(a,b,c,d)$ is a diagonal matrix if and only if $a = c = d = 0$. From this observation it becomes clear that $\pi(x^{-1}y)$ is a scalar multiple of the identity if and only if $x^{-1}y = 1_G$. Hence, $x = y$ and thus $\pi(x) = \pi(y)$. This shows that the set $\{ \pi(x) : x \in G \}$ is linearly independent, hence a basis for $M_{2n}(\bbC)$ since $|G| = (2n)^2$.
    
    The fact that $\Ind_L^G\rho \simeq \pi$ follows by essentially the same reasoning as in \cref{prop:smallest non-stabilizer Clifford code}, so does the fact that $W = V^{Cl}(G,\pi,L,\rho)$ is a non-stabilizer Clifford code. This completes the proof.
\end{proof}

\section{Examples of non-Clifford weak stabilizer codes}\label{sec:non-Clifford weak stabilizer codes}

To find examples of non-Clifford weak stabilizer codes we want to use \cref{thm:detectable errors:Clifford code} to its full potential. The following example is an example of a code that is permutation invariant. We thank Ningping Cao for the idea that lead us towards this example.

\begin{proposition}\label{prop:permutation invariant non-Clifford code}
    Let $n \geq 2$ and define $G = (\bbZ_2 \times \bbZ_2)^n \rtimes S_n$ where $S_n$ acts on $(\bbZ_2 \times \bbZ_2)^n$ by permuting the copies of $\bbZ_2 \times \bbZ_2$. Consider the projectively faithful irreducible projective representation $\pi \colon G \to U((\bbC^2)^{\otimes n})$ defined by
    \begin{equation*}
        \pi((a_i,b_i)_{i=1}^{n},\tau) = (X^{a_1}Z^{b_1} \otimes \dots \otimes X^{a_n}Z^{b_n}) \, \tau,
    \end{equation*}
    where $X$ and $Z$ are defined in \cref{ex:Pauli group}. Then $W := V^{St}(G, \pi, S_n, 1)$ is a weak stabilizer code that is not a Clifford code.
\end{proposition}
\begin{proof}
    Note that $\pi$ is indeed a projectively faithful irreducible projective representation by \cref{prop:permutation product PEM}. Let $\ket{0} = (1,0),\ket{1} = (0,1)$ denote the standard basis vectors on $\bbC^2$. We will use the shorthand notation
    \begin{equation*}
        \ket{j_1 j_2 \dots j_n} := \ket{j_1} \otimes \ket{j_2} \otimes \dots \otimes \ket{j_n},
    \end{equation*}
    with each $j_k \in \{0,1\}$, to denote the standard basis for $(\bbC^2)^{\otimes n}$. For each $k = 0,1,\dots,n$ we define the set
    \begin{equation*}
        B_k = \{ \ket{j_1 j_2 \dots j_n} \colon \sum_{i = 1}^{n} j_i = k \},
    \end{equation*}
    and define the vector
    \begin{equation*}
        \ket{D^n_k} = \sum_{\xi \in B_k} \xi.
    \end{equation*}
    One easily confirms that $W = \Span\{ \ket{D^n_k} \colon k = 0,1,\dots,n \}$, hence $W$ is in fact a weak stabilizer code as it is a nonzero subspace. If $W$ is a Clifford code, we have by \cref{cor:order of inertia group for Clifford codes} that
    \begin{equation*}
        |L_{(G,\pi)}(W)| 
        = \frac{\dim W}{\dim((\bbC^2)^{\otimes n})}|G|
        = 2^n \cdot (n+1)!
    \end{equation*}
    Suppose that $((a_i,b_i)_{i = 1}^{n},\tau) \in L_{(G,\pi)}(W)$. Then we have that
    \begin{equation*}
        \pi((a_i,b_i)_{i = 1}^{n},\tau)\ket{D_0^n}
        = \ket{a_1a_2 \dots a_n}
    \end{equation*}
    must belong to $W$. However, this is the case if and only if $a_1 = a_2 = \dots = a_n$. Thus, we get  that
    \begin{equation*}
        |L_{(G,\pi)}(W)| \leq 2^{n+1} \cdot n! < 2^n \cdot (n+1)!
    \end{equation*}
    whenever $2 < n + 1$. Hence, whenever $n \geq 2$, the order of $L_{(G,\pi)}(W)$ is too small for $W$ to be a Clifford code, completing the proof.
\end{proof}

These examples of non-Clifford weak stabilizer codes are very badly behaved. For example, when $n = 2$ the vector
\begin{equation*}
    (X \otimes I_2) \ket{D^2_0} = \ket{10}
\end{equation*}
is not in the code space or orthogonal to it. Hence, the assumptions in \cref{prop:detectable errors in terms of inertia group and stabilizers} are not satisfied, so we cannot describe the detectable errors in terms of the stabilizers and logical operations on the code.

\section{Product codes}\label{sec:Product codes}

In this section we show that the tensor product of two Clifford codes is again a Clifford code. In the special case where we only consider projective error models $(G,\pi) \in \PEM_V$ where $|G| = (\dim V)^2$, we show that if either of the Clifford codes we started with was not a weak stabilizer code, then the product code is also not a weak stabilizer code.

\begin{proposition}\label{prop:product code}
    Let $V_1$ and $V_2$ be Hilbert spaces, and suppose that $(G_1,\pi_1) \in \PEM_{V_1}$, and $(G_2,\pi_2) \in \PEM_{V_2}$. If $W_1 \subset V_1$, and $W_2 \subset V_2$ are Clifford codes, then $W_1 \otimes W_2$ is a Clifford code in $V_1 \otimes V_2$. Furthermore,
    \begin{align*}
        L_{(G_1 \times G_2, \pi_1 \otimes \pi_2)}(W_1 \otimes W_2) & = L_{(G_1, \pi_1)}(W_1) \times L_{(G_2, \pi_2)}(W_2), \text{ and } \\
        S_{(G_1 \times G_2, \pi_1 \otimes \pi_2)}(W_1 \otimes W_2) & = S_{(G_1, \pi_1)}(W_1) \times S_{(G_2, \pi_2)}(W_2).
    \end{align*}
\end{proposition}
\begin{proof}
    By \cref{prop:product error models} we have that $(G_1 \times G_2, \pi_1 \otimes \pi_2) \in \PEM_{V_1 \otimes V_2}$, hence the above statements are well-defined. Write $W_1 = V^{Cl}(G_1,\pi_1,L_1,\rho_1)$ and $W_2 = V^{Cl}(G_2,\pi_2,L_2,\rho_2)$. To confirm that $W_1 \otimes W_2$ is a Clifford code we want to check that
    \begin{equation*}
        \Ind_{L_1 \times L_2}^{G_1 \times G_2}(\rho_1 \otimes \rho_2) \simeq (\Ind_{L_1}^{G_1}\rho_1) \otimes (\Ind_{L_2}^{G_2}\rho_2).
    \end{equation*}
    Since $\Ind_{L_1}^{G_1}\rho_1 \simeq \pi_1$ and $\Ind_{L_2}^{G_2}\rho_2 \simeq \pi_2$, this would prove the desired isomorphism. Define $T \colon \Ind_{L_1 \times L_2}^{G_1 \times G_2}(W_1 \otimes W_2) \to (\Ind_{L_1}^{G_1}(W_1)) \otimes (\Ind_{L_2}^{G_2}(W_2))$ on simple tensors by
    \begin{equation*}
        T((x_1,x_2)\otimes (\xi_1 \otimes \xi_2))
        = (x_1 \otimes \xi_1) \otimes (x_2 \otimes \xi_2).
    \end{equation*}
    This is clearly an injective intertwiner. Since $T$ is an injective intertwiner between finite dimensional vector spaces of the same dimension, we get that $T$ is an isomorphism. This shows that $W_1 \otimes W_2$ is a Clifford code, and that
    \begin{equation*}
        L_{(G_1 \times G_2, \pi_1 \otimes \pi_2)}(W_1 \otimes W_2) = L_1 \times L_2 = L_{(G_1, \pi_1)}(W_1) \times L_{(G_2, \pi_2)}(W_2)
    \end{equation*}
    by \cref{thm:detectable errors:Clifford code} and \cref{cor:Clifford codes has unique inertia group}. Similarly, we get that
    \begin{equation*}
        S_{(G_1 \times G_2, \pi_1 \otimes \pi_2)}(W_1 \otimes W_2) = \ker(q\circ(\rho_1 \otimes \rho_2)) = S_{(G_1, \pi_1)}(W_1) \times S_{(G_2, \pi_2)}(W_2),
    \end{equation*}
    since $\rho_1(x_1) \otimes \rho_2(x_2)$ is a scalar multiple of the identity if and only if both $\rho_1(x_1)$ and $\rho_2(x_2)$ are scalar multiples of the identity. This completes the proof.
\end{proof}

\begin{corollary}\label{cor:non-stabilizer Clifford product code}
    Let $V_1$ and $V_2$ be Hilbert spaces, and suppose that $(G_1,\pi_1) \in \PEM_{V_1}$, and $(G_2,\pi_2) \in \PEM_{V_2}$ are such that $|G_1| = (\dim V_1)^2$, and $|G_2| = (\dim V_2)^2$. If $W_1 \subset V_1$ and $W_2 \subset V_2$ are Clifford codes such that $W_1$ is not a weak stabilizer code, then $W_1 \otimes W_2$ is a non-stabilizer Clifford code.
\end{corollary}
\begin{proof}
    By \cref{prop:product code} we have that $W_1 \otimes W_2$ is a Clifford code, and by \cref{thm:weak stabilizer Clifford code} we have that $W_1 \otimes W_2$ is also a weak stabilizer code if and only if 
    \begin{equation*}
        |G_1 \times G_2| = |L_{(G_1 \times G_2, \pi_1 \otimes \pi_2)}(W_1 \otimes W_2)|\cdot|S_{(G_1 \times G_2, \pi_1 \otimes \pi_2)}(W_1 \otimes W_2)|.
    \end{equation*}
    By \cref{prop:product code} we can rewrite this equation as
    \begin{equation*}
        |G_1|\cdot|G_2| = |L_{(G_1, \pi_1)}(W_1)|\cdot|L_{(G_2, \pi_2)}(W_2)|\cdot|S_{(G_1, \pi_1)}(W_1)|\cdot|S_{(G_2, \pi_2)}(W_2)|.
    \end{equation*}
    But, by \cref{thm:weak stabilizer Clifford code} and \cref{prop:failure of being a weak stabilizer code} we have that
    \begin{align*}
        |L_{(G_1, \pi_1)}(W_1)|\cdot|S_{(G_1, \pi_1)}(W_1)| & < |G_1|, \text{ and } \\
        |L_{(G_2, \pi_2)}(W_2)|\cdot|S_{(G_2, \pi_2)}(W_2)| & \leq |G_2|.
    \end{align*}
    Hence, $W_1 \otimes W_2$ is a non-stabilizer Clifford code, completing the proof.
\end{proof}

This results allows us to arbitrarily combine the codes from \cref{prop:family of non-stabilizer Clifford codes} to produce more examples of non-stabilizer Clifford codes.

\section{Further questions and remarks}\label{sec:further questions}

We have spent some effort in discussing when a Clifford code is a weak stabilizer code, and visa versa. \cref{thm:weak stabilizer Clifford code} gives a description in the case where $(G, \pi) \in \PEM_V$ is a projective error model such that $|G| = (\dim V)^2$. A natural question is if such a result is possible in the more general case. To be precise, we state the following question.

\begin{question}\label{Q1}
    Let $V$ be a Hilbert space and $(G,\pi) \in \PEM_V$. If $W \subset V$ is a Clifford code, is it possible to determine if $W$ is also a weak stabilizer code by knowing the orders $|L_{(G,\pi)}(W)|$, $|S_{(G,\pi)}(W)|$, and $|G|$?
\end{question}

It is also natural to wonder if all codes that are a Clifford code and a weak stabilizer code is automatically a stabilizer code. In other words, is
\begin{equation*}
    \{\text{Stabilizer codes}\} = \{\text{Clifford codes}\} \cap \{\text{Weak stabilizer codes}\}?
\end{equation*}
To state this question more formally we introduce the following notation. Consider a Hilbert space $V$ and a projective error model $(G,\pi) \in \PEM_V$. Define 
\begin{align*}
    St_{(G,\pi)}(V) & = \{ W \subset V : W \text{ is a stabilizer code with respect to } (G,\pi) \}, \\
    wSt_{(G,\pi)}(V) & = \{ W \subset V : W \text{ is a weak stabilizer code with respect to } (G,\pi) \}, \\
    Cl_{(G,\pi)}(V) & = \{ W \subset V : W \text{ is a Clifford code with respect to } (G,\pi) \}.
\end{align*}

We have been unable to produce an example of a code that is both a Clifford code and a weak stabilizer code but not a stabilizer code. We therefore state the following question.

\begin{question}\label{Q2}
    Does there exist a Hilbert space $V$ and a projective error model $(G,\pi) \in \PEM_V$ such that
    \begin{equation*}
        St_{(G,\pi)}(V) \neq wSt_{(G,\pi)}(V) \cap Cl_{(G,\pi)}(V)?
    \end{equation*}
\end{question}

In the special case where $|G| = (\dim V)^2$, we can restate this question as follows using \cref{thm:weak stabilizer Clifford code}.

\begin{question}\label{Q3}
    Does there exist a Hilbert space $V$, a subspace $W \subset V$, and a projective error model $(G,\pi) \in \PEM_V$ with $|G| = (\dim V)^2$, such that the equation
    \begin{equation*}
        |G| = |L_{(G,\pi)}(W)|\cdot|S_{(G,\pi)}(W)|
    \end{equation*}
    holds, but $S_{(G,\pi)}(W)$ is not normal in $G$?
\end{question}

Finally, it is worth stating that there are some aspects of quantum error correction that we have not discussed but would nevertheless fit into this framework. One of these aspects is the notion of \emph{code distance}. To define code distance, one would need to consider projective error models of the form $(G^n,\pi^{\otimes n})$, as explained in \cref{prop:product error models}. See \cite[Section 5]{GrasslMartínezNicolás10} where they discuss \emph{Hamming weight} and Clifford codes that can detect errors of Hamming weight $1$.

Another type of code that we could have discussed is \emph{subsystem codes}, see \cite{KribsLaflammePoulin2005}. If $V$ is a Hilbert space and $W \subset V$ is a subspace such that $W = W_1 \otimes W_2$ for some Hilbert spaces $W_1$ and $W_2$, we say that $W_1$ is a subsystem code. In this case one would be interested in the following set of detectable errors.
\begin{equation*}
    D^{sub}_{(G,\pi)}(W_1) = \{ x \in G \colon P_W \pi(x) P_W = \id_{W_1} \otimes T_x \text{ for some } T_x \in \calB(W_2) \}.
\end{equation*}
Note that if $W = V^{Cl}(G,\pi,I,\rho)$ is a Clifford code, $L = L_1 \times L_2$, and $\rho \simeq \rho_1 \otimes \rho_2$ for some irreducible projective representations $\rho_i \colon L_i \to U(W_i)$, $i = 1,2$, then $W \simeq W_1 \otimes W_2$. This gives a natural construction of a subsystem code from a Clifford code. With this, one could attempt to prove a result that is analogous to \cref{thm:detectable errors:Clifford code}.

\subsection*{Data availability}

Data sharing is not applicable to this article as no datasets were generated or analyzed during the current study.

\subsection*{Conflict of interest}

The author has no conflicts of interest to declare that are relevant to the content of this article.

\printbibliography

\end{document}